\documentclass[journal]{IEEEtran} 		
\usepackage{amsmath,amsfonts}

\usepackage{array}
\usepackage[caption=false,font=normalsize,labelfont=sf,textfont=sf]{subfig}
\usepackage{textcomp}
\usepackage{stfloats}
\usepackage{url}
\usepackage{verbatim}
\usepackage{graphicx}
\usepackage{siunitx} 			
\usepackage{multirow} 		
\usepackage{booktabs}	 
\usepackage{ragged2e}
\usepackage[most]{tcolorbox}

\def\BibTeX{{\rm B\kern-.05em{\sc i\kern-.025em b}\kern-.08em
    T\kern-.1667em\lower.7ex\hbox{E}\kern-.125emX}}

\usepackage{balance}

\usepackage[english]{babel}  

\usepackage{amssymb}
\usepackage{amsmath}
\usepackage{graphicx}					
\usepackage{hyperref} 

\usepackage{amsthm}				

\usepackage{pdfpages}				
\usepackage{dsfont}   					

\usepackage{algorithm}
\usepackage{algpseudocode}

\usepackage{siunitx}  

\newtheoremstyle{mystyle}
{3pt}{3pt}        
{\normalfont}     
{}                
{\bfseries\itshape}       
{.}               
{.5em}            
{}                

\theoremstyle{mystyle}

\newtheorem{Thm}{Theorem}

\newtheorem{Rmk}{Remark}
\newtheorem{Def}{Definition}

\newcommand{\R}{\mathbb{R}}

\newcommand{\eeq}{\end{equation}}

\newcommand{\bX}{\mbox{\boldmath $X$}}
\newcommand{\bh}{\mbox{\boldmath $h$}}

\newcommand{\bc}{\mbox{\boldmath $c$}}

\newcommand{\bgamma}{\mbox{\boldmath $\gamma$}}

\newcommand{\bGamma}{\mbox{\boldmath $\Gamma$}}

\newcommand{\bG}{\mbox{\boldmath $G$}}

\newcommand{\bx}{\mbox{\boldmath $x$}}

\newcommand{\bg}{\mbox{\boldmath $g$}}

\newcommand{\bv}{\mbox{\boldmath $v$}}

\newcommand{\bphi}{\mbox{\boldmath $\phi$}}
\newcommand{\bI}{\mbox{\boldmath $I$}}
\newcommand{\bbtheta}{\mbox{\boldmath $\theta$}}

\newcommand{\bSigma}{{\bf \Sigma}}

\newcommand{\bt}{\mbox{\boldmath $t$}}

\newcommand{\by}{\mbox{\boldmath $y$}}
\newcommand{\ds}{\displaystyle}

\newcommand{\beq}{\begin{equation}}

\newcommand{\tr}      {{\mathrm{tr}}}    
\newcommand{\E}       {{\mathrm E}}      

\newtheorem{proposition}{Proposition}

\renewcommand{\E}       {{\mathbb{E}}}      

\begin{document}
	
\title{A Stochastic Optimization Framework for RIS-Aided Wireless Network Design}
\author{
	Davide Gagliardi,  Alessio Zappone, {\em Fellow, IEEE},\\ Domenico Ciuonzo, {\em Senior Member, IEEE}, Marco Di Renzo, {\em Fellow, IEEE}
	\thanks{D. Gagliardi and A. Zappone are with the University of Cassino and Southern Lazio, Italy (\{davide.gagliardi, alessio.zappone\}@unicas.it). D. Ciuonzo is with the University of Napoli Federico II, Italy (domenico.ciuonzo@unina.it). M. Di Renzo is with CNRS and CentraleSup\'elec, Institute of Electronics and Digital Technologies (IETR), Avenue de la Boulaie, 35576 Cesson-S\'evign\'e, France (marco.direnzo@centralesupelec.fr), and with King's College London, Department of Engineering - Centre for Telecommunications Research, WC2R 2LS London, United Kingdom (marco.di\_renzo@kcl.ac.uk). 
}}

\maketitle

\begin{abstract}
Reconfigurable intelligent surfaces (RISs) are a promising technology for improving the spectral and energy efficiency of future wireless networks, which make use of metasurfaces. However, optimizing RIS configurations typically leads to large-scale, non-convex problems whose complexity grows significantly with the number of scattering elements and the adoption of advanced metasurface architectures. In this paper, we develop a stochastic optimization framework for RIS-aided wireless networks based on continuous versions of the ($a$) cross-entropy (CE) and ($b$) Metropolis-Hastings (MH) methods. Unlike existing stochastic approaches that mainly focus on discrete optimization, the proposed framework directly handles continuous variables and can be readily applied to discrete settings through relaxation and projection. We provide a theoretical characterization of the proposed algorithms, including convergence guarantees and efficiency analysis. 
The framework is applied to ($i$) achievable-rate maximization with nearly-passive RISs and ($ii$) energy-efficiency maximization with active RISs. Numerical results show that the proposed methods achieve performance comparable to, or better than, state-of-the-art deterministic algorithms, while reducing execution times up to 10 times in representative scenarios.
\end{abstract}


\section{Introduction}
\IEEEPARstart{T}{he sixth generation (6G)} of wireless networks, expected to be deployed around 2030, is envisioned to support unprecedented traffic demands and Tb/s data rates while ensuring sustainable energy consumption and implementation complexity~\cite{6GMagazine}.
Among the technologies envisioned for 6G, reconfigurable intelligent surfaces (RISs) have emerged as a promising solution for improving both spectral and energy efficiency~\cite{huang2019reconfigurable,di2020smart}. 
RISs are metasurfaces that can be deployed either to create new and controllable propagation paths~\cite{strinati2021reconfigurable,zhang2025smart,katwe2024overview} or to enable low-complexity beamforming architectures integrated into the transceivers~\cite{huang2020holographic,an2023tutorial,gong2024holographic_Tut}.
In either scenario, a decisive advantage of metasurfaces is the fact that they fully operate in the \emph{wave domain}, dispensing with digital signal processing, analog-to-digital and digital-to-analog converters, which drastically reduces their hardware complexity, cost, and energy consumption. 
In turn, this provides large array gains, because the number of meta-atoms that can be equipped on a single metasurface can be significantly higher than the number of antennas in traditional arrays based on digital processing. This has been shown to provide larger data-rate and energy efficiency (EE) compared to conventional multiple-antenna technologies~\cite{ZapTWC2021,guo2023green,gong2024holographic}.

However, achieving the performance gains promised by metasurfaces requires the \emph{optimization of their scattering coefficients}, since non-optimized configurations can significantly degrade system performance. This poses a major complexity issue, as relevant deployments often involve \emph{hundreds of controllable meta-atoms}, resulting in large-scale optimization problems, whose complexity grows rapidly with the size of the metasurface~\cite{Nerini2024_MultiPort,Zhu25Movable,Fotock_RHB,Mursia2025_T3DRIS}. In addition, emerging metasurface architectures such as beyond-diagonal metasurfaces, which equip more meta-atoms than diagonal ones~\cite{Soleymani_BDRIS,Li2024_BeyondDiag} and stacked intelligent metasurfaces (SIMs), which are composed of multiple layers of meta-atoms~\cite{Magbool2025_SFIM,An2024_SIM,An2025_SIM}, further increase the number of optimization variables and their mutual coupling, resulting in even higher computational complexity.
Consequently, even optimization algorithms with polynomial complexity may become impractical in large-scale deployments or rapidly varying wireless environments.
Furthermore, typical optimization problems aimed at the maximization of common network performance metrics, e.g. achievable rate, EE, latency, localization accuracy, are \emph{non convex}.
Existing approaches therefore often rely on iterative optimization frameworks, which require solving a sequence of auxiliary (convex) optimization problems. While effective, these methods generally provide \emph{locally-optimal solutions} and may incur a substantial \emph{computational burden}, especially in large-scale metasurface deployments or rapidly-varying wireless environments, in which the optimization of the RIS coefficients must be updated frequently. 

An optimization framework that can overcome these drawbacks is Monte Carlo (MC) optimization, a branch of stochastic optimization that employs probabilistic techniques to explore the feasible set. In MC optimization, the objective and constraint functions are deterministic, but the search for the optimum is performed by repeatedly sampling the feasible set according to a probability distribution. The probability distribution is updated through an iterative algorithm whose output is ideally a distribution which allows sampling the optimal solution with probability one. 
As such, MC methods can be applied to non-convex optimization problems without requiring convexity assumptions or the evaluation of objective-function gradients. 
Among the most prominent approaches are the cross-entropy (CE) method~\cite{rubinstein2004cross,rubinstein1999cross} and the Metropolis-Hastings (MH) method~\cite{hastings1970monte,metropolis1953equation}. 
The former iteratively refines a sampling distribution toward regions containing high-quality solutions by minimizing the Kullback–Leibler divergence from an ideal target distribution~\cite{rubinstein2004cross,rubinstein1999cross}, whereas the latter explores the feasible set through a Markov-chain-based random process~\cite{hastings1970monte,metropolis1953equation}.

\noindent
\textbf{Related Work:}
In the context of optimization for metasurface-aided wireless networks, relatively \emph{few works have explored stochastic optimization techniques}. 
In~\cite{ren2022configuring,Xiao2025,Ross2024,Chen2024,Zhang2025}, stochastic optimization is employed for the optimization of single-user communication systems aided by nearly-passive metasurfaces.
In~\cite{ren2022configuring}, a single-user single-input-single-output (SISO) system aided by a nearly-passive RIS with discrete phase shifts is considered, and ad-hoc stochastic optimization methods are developed for signal-to-noise-ratio (SNR) maximization based on received signal samples. 
A similar SISO setup is studied in~\cite{Xiao2025}, where the CE method is used to optimize the positions of the discrete phase shifts of a fluid RIS for capacity maximization by the CE method.
In~\cite{Ross2024}, genetic programming is used to optimize the discrete phase shifts in an RIS-aided multiple-input-single-output (MISO) system, whereas~\cite{Chen2024} tackles the same problem by the CE method.
Finally,~\cite{Zhang2025} applies the CE method to discrete phase-shift optimization of system capacity in an RIS-aided single-user multiple-input-multiple-output (MIMO) system.

In~\cite{JChen2024b,El-Meadawy2025,JChen2026,Shen26,JChen2024c,JChen2024a} the downlink of a multi-user multiple-input single-output (MISO) system is considered. 
In~\cite{JChen2024b}, zero-forcing precoding is employed and the discrete phase shifts of a nearly-passive simultaneous transmit and receive (STAR) reconfigurable metasurface are optimized for sum-rate maximization by the CE method. 
In~\cite{El-Meadawy2025}, a nearly-passive SIM is employed, whose discrete phase shifts are optimized  by the CE and simulated annealing algorithms, for the optimization of the system sum-rate and minimum rate. In~\cite{JChen2026}, a nearly-passive metasurface with movable elements is considered and the discrete phase shifts and positions of the elements are optimized by the CE method, assuming zero-forcing transmission. 
A continuous version of the CE method is considered in~\cite{Shen26} for phase-shift optimization in a movable RIS-assisted non-orthogonal multiple access (NOMA) system toward sum-rate maximization.
System EE maximization via discrete CE-based optimization is investigated in~\cite{JChen2024c,JChen2024a}.
Finally, \cite{GeneticPan2022} considers an RIS-aided multi-user single-input multiple-output (SIMO) uplink system and employs genetic programming for sum-rate maximization via discrete shifts.

\noindent
\textbf{Research Gaps and Novel Contributions:}
Based on the aforementioned state-of-the-art, \emph{four main research gaps can be identified}.
First, with the exception of~\cite{Shen26}, existing stochastic optimization methods for RIS design \emph{operate directly on discrete feasible sets}. While this is motivated by the finite number of phase shifts of the meta-atoms employed in practical implementations, it often leads to large combinatorial search spaces. By contrast, continuous relaxations followed by projection onto the discrete feasible set have been shown to reduce the complexity of deterministic optimization methods. Extending this principle to stochastic optimization requires the development of MC algorithms capable of operating directly in continuous domains. Second, previous works on MC-based optimization for RIS design \emph{focus on the CE method, whereas the MH method has not been explored, yet.}
Third, all previous works apply MC optimization methods \emph{without formally discussing their convergence and optimality properties}. A mathematical treatment of the convergence and efficiency of MC optimization has not been addressed to date, especially for applications to continuous domains. 
Fourth, prior work has \emph{mainly focused on rate maximization problems in conventional nearly-passive RIS architectures} (in single-user or downlink multi-user systems). The application of stochastic optimization to more general metasurface models, such as active RISs with global reflection constraints, and to EE optimization problems remains comparatively unexplored. In order to fill these gaps, this works makes the following \emph{novel contributions}:

\begin{itemize}
\item We develop continuous-domain CE and MH algorithms for RIS optimization. The proposed framework directly handles continuous optimization variables and can also be applied to discrete RIS designs through relaxation and projection. Numerical results show that this approach achieves the same performance as discrete-domain optimization, but with a significantly lower complexity.

\item We demonstrate the versatility of the proposed framework through \emph{two representative case studies}: ($i$) achievable-rate maximization in nearly-passive RISs and ($ii$) (global) EE maximization in active RISs with global reflection constraints~\cite{DiRenzoGlobal}. These examples illustrate the applicability of the proposed algorithms to a broad class of non-convex RIS design problems.

\item We provide a rigorous theoretical analysis of the proposed algorithms, establishing convergence guarantees and characterizing the efficiency of the resulting solutions. To the best of our knowledge, this is the first convergence study of continuous CE- and MH-based optimization tailored to RIS-aided wireless networks.

\item We numerically assess the performance of the proposed methods, comparing them against state-of-the-art benchmarks, including the popular alternating optimization (AO) method. The results show that the proposed CE and MH algorithms perform better than existing approaches while significantly reducing the execution time.

\end{itemize}

The rest of the work is organized as follows. Sec.~\ref{sec:optimization_framework} provides the theoretical analysis of the continuous versions of the CE and MH algorithms. Sec.~\ref{sec:case_study} applies the CE and MH to the two case-studies. Numerical results are provided in Sec.~\ref{sec:numerical_results}, while concluding remarks are given in Sec.~\ref{sec:conclusions}. 

\section{Optimization framework}\label{sec:optimization_framework}
Consider the generic maximization problem 
\begin{align}\label{Prob:Problem_CE}
&\max_{\bx}\ds {\mathcal{J}}(\bx)\;,\;\textrm{s.t.}\;\bx\in{\cal X}\;. 
\end{align}
We assume only that ${\cal J}(\cdot)$ is continuous and ${\cal X}$ is compact, so that \eqref{Prob:Problem_CE} admits a solution. The CE and MH methods are described in  the next two subsections. 

\subsection{Cross-entropy maximization}
Traditionally, the CE method for optimization has been applied to discrete optimization problem. This section describes its application to continuous optimization problems, which is the focus of this work. As evaluated in Sec. \ref{sec:numerical_results}, the continuous formulation enjoys a lower complexity than its discrete counterpart.   
Consider Problem \eqref{Prob:Problem_CE} and denote by $\bx^{\star}$ a global maximizer of $\mathcal{J}$ over $\mathcal{X}$, and $y^{\star}=\mathcal{J}(\bx^{\star})$.  Rather than directly tackling \eqref{Prob:Problem_CE} by gradient-based methods or other deterministic approaches, the CE method tackles \eqref{Prob:Problem_CE} by considering the following stochastic relaxation
\begin{equation}\label{Eq:ProbStochastic}
p_{y}=\mathbb{P}\left( \mathcal{J}( \bX )\geq y \right) \geq \rho\;,
\end{equation}
wherein $\bX$ is a random variable with support over $\mathcal{X}$, $y$ is a desired performance level, and $\rho>0$. Thus, Problem \eqref{Prob:Problem_CE} is converted into finding a probability distribution for $\bX$, such that, with probability larger than $\rho$, the objective evaluated at a realization of $\bX$ is above a desired threshold $y$. Otherwise stated, we look for a distribution $g^{\star}$ such that a feasible point randomly extracted from $\mathcal{X}$ according to $g^{\star}$ fulfills \eqref{Eq:ProbStochastic}. If $y=y^{\star}$ and $\rho=1$, \eqref{Eq:ProbStochastic} would lead to Dirac delta functions centered at the global solutions of \eqref{Eq:ProbStochastic}. However, such a distribution can not be computed, since the value $y^{\star}$ is not known. Then, assume $y<y^{\star}$ and consider a tentative distribution $f$ over $\mathcal{X}$. Denoting by $\mathds{1}_{y}(\bx)$ the indicator function of $\{\mathcal{J}(\bx)\geq y\}$, i.e. $\mathds{1}_{y}(\bx)=1$ if $\mathcal{J}(\bx)\geq y$ and $0$ otherwise, we have that\footnote{The notation $\mathbb{E}_{f}[h(\bX)]$ means the expectation of the random quantity $h(\bX)$ with respect to the distribution $f$, i.e. $\int h(\bx)f(\bx)d\bx$.} 
\begin{equation}\label{Eq:Py}
p_{y}\!=\!\mathbb{P}\left( \mathcal{J}(\bX)\!\geq \!y\right)\!=\!\mathbb{E}_{f}\!\left[\!\mathds{1}_{y}(\bX)\!\right]\!=\!\mathbb{E}_{g}\!\left[\!\mathds{1}_{y}(\bX)\frac{f(\bX)}{g(\bX)}\!\right]\;,
\end{equation}
where $g$ is a new distribution over\footnote{In order for $g/q$ to be well-defined, $q$ must be absolutely continuous with respect to $g$, i.e. the support of $g$ must be included in the support of $q$, so that all points at which $q=0$ need not be considered in the ratio.} $\mathcal{X}$. The solution of \eqref{Eq:ProbStochastic} is a (properly normalized) distribution $g^{\star}$ that concentrates all probability density in the set $\{\mathcal{J}( \bX )\geq y\}$, i.e. \footnote{Note that quantity $\mathbb{E}_{f}\left[\mathds{1}_{y}(\bX)\right]$ acts as the normalization constant to ensure that the distribution integrates to one.}
\begin{equation}\label{Eq:g_star}
g^{\star}(\bx)=\frac{\mathds{1}_{y}(\bx)f(\bx)}{\mathbb{E}_{f}\left[\mathds{1}_{y}(\bX)\right]}=\frac{\mathds{1}_{y}(\bx)f(\bx)}{p_{y}}\;.
\end{equation}
Unfortunately, \eqref{Eq:g_star} can not be computed either, because the normalization constant $\mathbb{E}_{f}\left[\mathds{1}_{y}(\bX)\right]$ is also unknown. However, the CE method provides a framework to approximate \eqref{Eq:g_star}, in the sense of the Kullback-Leibler (KL) divergence, also known as relative entropy. Specifically, the CE method looks for a distribution $q$ that minimizes the KL-divergence from the optimal distribution $g^{\star}$ in \eqref{Eq:g_star}. The KL-divergence between $g^{\star}$ and $q$ can be written as\footnote{For simplicity, and without loss of generality, we use natural logarithms.}
\begin{align}
&D(g^{\star}||q)\!=\!\mathbb{E}_{g^{\star}}\!\!\left[\ln\frac{g^{\star}(\bX)}{q(\bX)}\right]\!=\!\mathbb{E}_{g^{\star}}\left[\ln{g^{\star}(\bX)}\right]-\mathbb{E}_{g^{\star}}\left[\ln{q(\bX)}\right]\notag\\
&=\int\frac{\mathds{1}_{y}(\bx)f(\bx)}{p_{y}}\left(\ln\left(\frac{\mathds{1}_{y}(\bx)f(\bx)}{p_{y}}\right)-\ln\left(q(\bx)\right)\right)d\bx\notag\\
&=\!\frac{1}{p_{y}}\mathbb{E}_{f}\!\left[\mathds{1}_{y}(\!\bX\!)\ln(f(\!\bX\!))\right]\!-\!\ln(p_{y})\!-\!\frac{1}{p_{y}}\mathbb{E}_{f}\!\left[\mathds{1}_{y}(\bX)\ln(q(\bX))\right]\notag
\end{align}
Therefore, the \emph{minimization} of $D(g^{\star}||q)$ with respect to $q$ is equivalent to the following \emph{maximization} problem
\begin{align}\label{Prob:Opt_q}
\max_{q}\mathbb{E}_{f}\left[\mathds{1}_{y}(\bX)\ln(q(\bX))\right]
\end{align}
Due to the continuous nature of Problem \eqref{Prob:Opt_q}, the optimization is with respect to a function, which may be computationally challenging in general. Then, it is convenient to restrict the search for $q$ to a specific family of distributions, whose parameters become the optimization variables \cite[Sec. 2.1]{deboer2005tutorial}. For example, if $f$ and $q$ are assumed Gaussian, the optimization reduces to determining the mean vector and covariance matrix of $q$. In general, assume that $f$ and $q$ belong to a family of distributions characterized by the parameter vectors $\bbtheta_{f}$ and $\bbtheta_{q}$, respectively. Then, Problem \eqref{Prob:Opt_q}  can be recast as 
\begin{align}\label{Prob:Opt_theta}
\max_{\bbtheta_{q}}\mathbb{E}_{\footnotesize\bbtheta_{f}}\left[\mathds{1}_{y}(\bX)\ln(q(\bX,\bbtheta_{q})\right]
\end{align}
The following definition introduces the family of distributions that are canonically used in the CE method. 
\begin{Def}
A random vector $\bX\in\R^N$ is said to belong to the \textit{$m$-dimensional exponential family} ($m$-EF) if its probability density function (PDF) can be expressed as: 
	\begin{equation}
		\label{def:NEF_form}
		q(\bx,\boldsymbol{\theta}) := c(\boldsymbol{\theta}) e^{\boldsymbol{\theta}^T \bt(\bx)} h(\bx) 
	\end{equation}
	where $\bt(\bx) = [t_1(\bx), \dots, t_m(\bx)]^T$, $h(\bx) > 0$, and
	$c(\bbtheta) = \left( \int {e^{\bbtheta^{T}\bt(\bx)}}  h(\bx) d\bx \right)^{-1} < \infty$ 
	$\forall \bbtheta \in \R^m$. \\
	If $\bt(\bx)=\bx$ (which also implies $m=N$), then  $\bX$ is said to belong to the \textit{natural exponential family} (NEF).
\end{Def}
\noindent The $m$-EF includes a wide variety of distributions, such as the \textit{exponential}, \textit{Gaussian}, \textit{Gamma}, \textit{Beta}  etc. and their multivariate versions \cite{rubinstein2004cross}, \cite[Sec. 2.3]{kroese_rubinstein_Glynn_2013_CE_meth_for_estim}. For example, $X\sim {\cal N}(\mu,\sigma^{2})$ belongs to the $m$-EF with $m=2$, $\bbtheta^{T}=[\mu/\sigma^{2}\;, -1/(2\sigma^{2})]$, $\bt(x)=[x\;, x^{2}]^{T}$, $c(\mu,\sigma)=e^{-\mu^{2}/(2\sigma^{2})}/\sqrt{2\pi}\sigma$, and $h(x)=1$. This simple example also shows that that the mean and variance can be obtained as  $\mu=\E[t_{1}(x)]$ and $\sigma^{2}=\E[t_{2}(x)]-(\E[t_{1}(x)])^{2}$. Similar steps can be taken for any distribution belonging to the $m$-EF, showing that a one-to-one correspondence exists between the parameters of the distribution, e.g. mean and variance, and the components of $\E_{f}[\bt(\bX)]$. As a consequence, the problem of finding the optimal $\bbtheta$ is equivalent to finding $\E_{f}[\bt(\bx)]$. Nevertheless, no general formula exists that maps the relationship between $\bbtheta$ and $\E_{f}[\bt(\bX)]$. Thus, the map between $\bbtheta$ and $\E_{f}[\bt(\bX)]$ must be derived for the specific distribution of choice. When $q(\bx, \boldsymbol{\theta})$ belongs to the $m$-EF, the following proposition holds. 
\begin{proposition}\label{Prop:Opt_theta}
Assume $q(\bx,\boldsymbol{\theta})$ belongs to the $m$-EF. Then $\mathbb{E}_{f}\left[\mathds{1}_{y}(\bX)\ln(q(\bX, \bbtheta))\right]$ is concave in $\bbtheta$. 
\end{proposition}
\begin{IEEEproof}
See Appendix.
\end{IEEEproof}
\noindent As a consequence of Proposition \ref{Prop:Opt_theta}, Problem \eqref{Prob:Opt_theta} is convex whenever $q$ belongs to the $m$-EF, and can be solved by setting the gradient of the objective function to zero, which yields
\begin{equation}\label{Eq:OptParameters}
\E_{f}[\bt(\bX)]=\frac{\E_{f}[\mathds{1}_{y}(\bX)\bt(\bX)]}{\E_{f}[\mathds{1}_{y}(\bX)]}.
\end{equation}
Given the one-to-one map between $\bbtheta$ and $\E_{f}[\bt(\bX)]$, computing \eqref{Eq:OptParameters} yields the optimal solution of \eqref{Prob:Opt_theta}. Finally, in practice, computing \eqref{Eq:OptParameters} might be difficult due to the statistical expectations. For this reason, it is customary to approximate it  through sample averages, namely
\begin{equation}\label{Eq:OptParametersMC2}
\E_{f}[\bt(\bX)]\approx\frac{\sum_{k=1}^{K}\mathds{1}_{y}(\bx_{k})\bt(\bx_{k})}{\sum_{k=1}^{K}\mathds{1}_{y}(\bx_{k})}=\frac{1}{|\mathcal{E}|}\sum_{\scriptsize{\bx}_{k}\in\mathcal{E}}\bt(\bx_{k})\;,
\end{equation}
with $\bx_{1},\ldots,\bx_{K}$ \emph{feasible} points drawn from $\mathcal{X}$ according to the distribution $q$ and $\mathcal{E}=\{\bx_{k} | \mathcal{J}(\bx_{k})\geq y\}$ containing only the samples that yield an objective value larger than $y$. The set $\mathcal{E}$ is called the \emph{elite set}, and its elements the \emph{elite samples}. Finally, the procedure outlined above finds a distribution that fulfills the condition in \eqref{Eq:ProbStochastic}, for a fixed threshold $y$. Clearly, for optimization purposes, $y$ should be as large as possible, but, of course, it can not be larger than the optimal solution of \eqref{Prob:Problem_CE}, $y^{\star}$. Since, in practice, $y^{\star}$ is unknown, the CE method starts from $y^{[0]}\leq y^{\star}$, which can be found by simply computing the objective function at any feasible point. Then, the CE method progressively increases the threshold $y^{[t]}$ in each iteration $t$, by finding a distribution that fulfills \eqref{Prob:Problem_CE} with $y=y^{[t]}$. The overall procedure is outlined in 
Algorithm \ref{algo:basic}, which constructs a sequence of density parameters $\{\boldsymbol{\theta}^{[t]}, t\geq0\}$ and a \textit{non-decreasing} sequence of levels $\{y^{[t]}, t\geq0\}$, which both converge \cite{homem_et_rubinstain_2002estimation}, \cite{shapiro1996simulation}. Fig. \ref{fig:FigFlusso} shows the flow chart of the CE algorithm, while its pseudo-code is reported in \textbf{Algorithm} \ref{algo:basic}.  
\begin{figure}[!h]
 \includegraphics[width=0.5\textwidth]
    {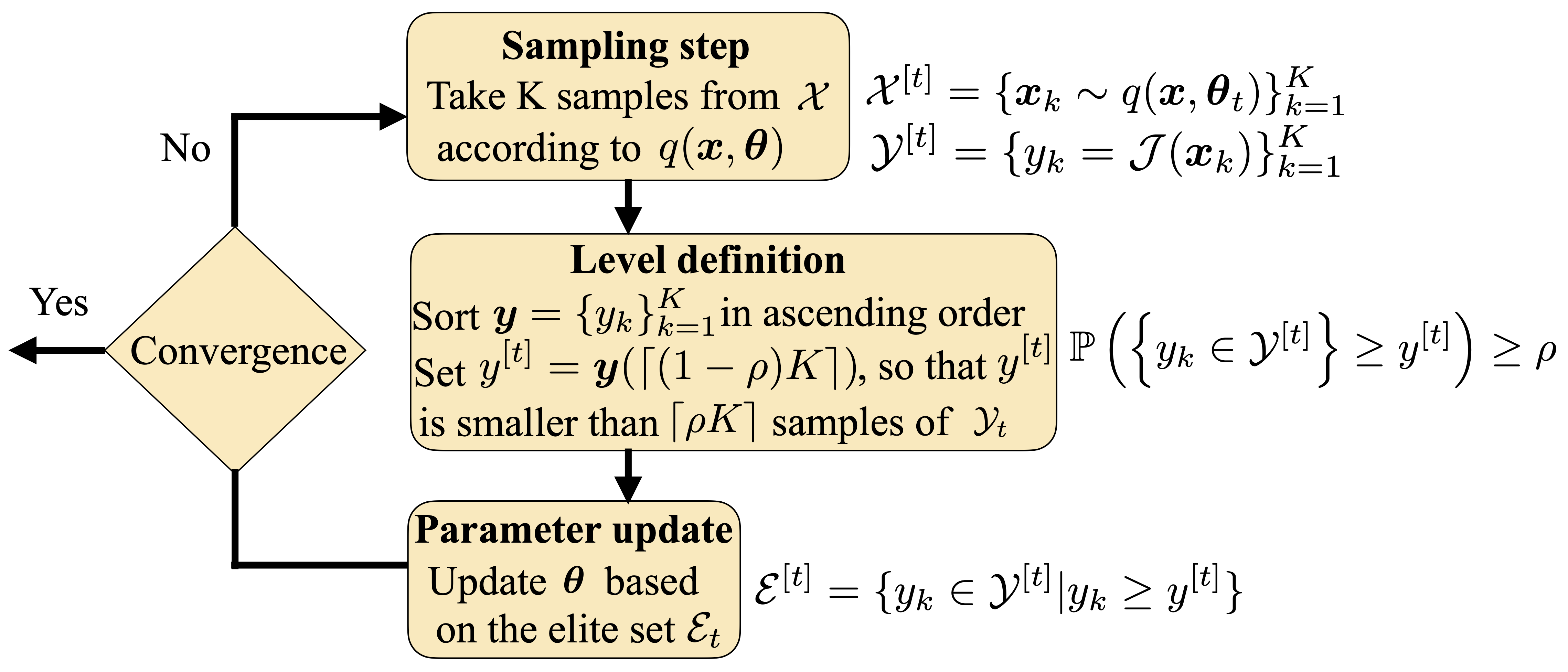}
    \caption{Flow chart of the CE maximization algorithm}
    \label{fig:FigFlusso}
\end{figure}

\begin{algorithm}
	\caption{CE Method for Maximization - Base version}
	\small
	\label{algo:basic}
	\begin{algorithmic}[1]
		\State \textbf{Input:} $\boldsymbol{\theta}^{[0]}$, $K$, $\rho$
		\State \textbf{Initialize:} 
		$t \gets 0$, 
		$\bx^\star  \gets \bx\in\mathcal{X}$, 
		$y^\star \gets \mathcal{J} \left( \bx^\star \right)$
		\Repeat
		\State \textbf{Generate} the sample set $\mathcal{X}^{[t]} = \{\bx_{k} \sim q(\bx; \bbtheta^{[t])} \}_{k=1}^K $
		\State \textbf{Compute} performances set $\mathcal{Y}^{[t]} =\left\lbrace y_k = \mathcal{J} \left( \bx_k\right)\right\rbrace_{k =1}^K$
		\State \textbf{Sort} $\mathcal{Y}^{[t]}$ in ascending order: $y_{(1)} \leq \dots \leq y_{(K)}$ 
	    \State \textbf{Maxim-um/izer update}\\ \quad \quad \textbf{if} $ y^\star< y_{(K)}$  \textbf{then} $y^\star\gets y_{(K)}$ $\mathbf{x}^\star\gets \bx_{(K)}$
		\State \textbf{Set} the threshold $y^{[t]} \gets y_{(\lceil (1-\rho) K \rceil)}$
		\State \textbf{Identify} the Elite Set: $\mathcal{E}^{[t]} = \{ \bx_{k}\in\mathcal{X}^{[t]} \mid y_k \geq y^{[t]} \}$
		\State \textbf{Compute}:  $\E_{\scriptsize\bbtheta^{[t]}}[\bt(\bX)] = \frac{1}{|\mathcal{E}^{[t]}|} \sum_{\bx \in \mathcal{E}^{[t]}} \bt(\bx)$
		\State \textbf{Update}: $\boldsymbol{\theta}^{[t+1]}$ based on $\E_{\scriptsize\bbtheta^{[t]}}[\bt(\bX)]$, $t \gets t + 1$
		\Until{\textit{Convergence}}
		\State \textbf{Output:} $\bbtheta^{\star}\gets\bbtheta^{[t]}$, $y^\star\gets y^{[t]}$, $\bx^\star\sim q(\bx, \bbtheta^{\star}) $
	\end{algorithmic}
\end{algorithm}
\noindent In each iteration, Algorithm \ref{algo:basic} applies the following steps:
\begin{enumerate}
	\item \textbf{Sampling step:} generate the sets $\mathcal{X}^{[t]}=\{\bx_{k}\}_{k=1}^{K}$ and $\mathcal{Y}^{[t]}=\{y_{k}=\mathcal{J}(\bx_{k})\}_{k=1}^{K}$ by sampling $K$ feasible points according to $q(\bx,\bbtheta^{[t]})$, and evaluating the corresponding objective values. 
	\item \textbf{Level definition:} set the level $y^{[t]}=y(\lceil (1-\rho) K\rceil)$, with $y(\lceil (1-\rho)  K\rceil)$ denoting the $\lceil(1-\rho)  K\rceil$-th largest element in $\mathcal{Y}^{[t]}$. This ensures that $\mathbb{P}(Y_{k}\geq y^{[t]})\geq \rho$, with $Y_{k}$ a random element from $\mathcal{Y}^{[t]}$. 	
	\item \textbf{Parameter update:} Define the \emph{elite set} at iteration $t$ as the set $\mathcal{E}^{[t]}=\{\bx_{k}\in\mathcal{X}^{[t] }| \mathcal{J}(\bx_{k})\geq y^{[t]}\}$ Then, compute \eqref{Eq:OptParametersMC2} based on $\mathcal{E}^{[t]}$, and update the density parameter vector $\boldsymbol{\theta}^{[t+1]}$ accordingly. Thus, $\mathcal{E}^{[t]}$ contains exactly $\lceil\rho K\rceil$ samples and the distribution parameters are updated based only on the elite samples. This iteratively refines the sampling distribution, increasing the threshold value $y^{[t]}$ for the considered probability value $\rho$, eventually fulfilling \eqref{Eq:ProbStochastic} with the highest possible threshold $y$. 
\end{enumerate}
\begin{Rmk}\label{Rmk:WrappedTruncated}
	If a pdf $q(\mathbf{x}, \boldsymbol{\theta})$ is restricted to a compact support $D$ (e.g., a truncated Gaussian), it remains a member of the $m$-EF by redefining the base measure as $\tilde{h}(\mathbf{x}) = h(\mathbf{x}) \mathds{1}_D(\mathbf{x})$. Since the exponential structure is preserved, the optimality condition $\nabla_{\boldsymbol{\theta}} \tilde{A}(\boldsymbol{\theta}) = \E_{\boldsymbol{\theta}}[\mathbf{t}(\mathbf{X})]$ still holds. In this case, the proof in Appendix carries over by defining $\tilde{A}(\boldsymbol{\theta}) = -\ln \tilde{c}(\boldsymbol{\theta})$, instead of $A(\boldsymbol{\theta}) = -\ln c(\boldsymbol{\theta})$. 
\end{Rmk}

Although Algorithm \ref{algo:basic} converges, its convergence speed can be improved by the following proposed modifications. 
\paragraph{Variable sample size $N$} this modification is motivated by the observation that, as the algorithm approaches convergence, the probability of finding elite samples may decrease, thus making the estimation of the parameters in $\boldsymbol{\theta}$ less reliable. This can be avoided by allowing the sample size $N$ \emph{to vary in each iteration}. Specifically, at each iteration, the algorithm validates the innovation of the lower bound ($y^{[t]} > y^{[t-1]}$) starting from a minimal sample size $N_0$. If the improvement is not detected under the current density $q(\mathbf{x}; \boldsymbol{\theta})$, a \textit{sub-iteration} is triggered, in which the number of samples $N^{[t]}$ is increased by a fixed increment $\Delta N$. This \textit{batch expansion} continues until either a successful innovation is verified or the pre-defined maximum sample budget $N_{max}$ is reached. Applying this approach, in each iteration the rarity threshold is implicitly set to $\rho = N_{elite}/N^{[t]}$, so that the cardinality of the elite set is fixed to $N_{elite}$. In this formulation, the sample $x_{(N^{[t]})}$ represents the best elite sample in iteration $t$, while the $N_{elite}$-th order statistic defines the $\rho^{[t]}$-quantile of the sample set.
\paragraph{Dynamic smoothing} Another modification that is adopted as a regularization of the algorithm, is the use of a smoothing update for the parameter vector. Namely, the new parameter vector is updated as a convex combination of the parameter vector $\tilde{\boldsymbol{\theta}}^{[t+1]}$ that results from \eqref{Eq:OptParametersMC2}, and the parameter vector at the previous iteration $\boldsymbol{\theta}^{[t]}$, namely $\boldsymbol{\theta}^{[t+1]} = (1-\alpha) \tilde{\boldsymbol{\theta}}^{[t+1]} + \alpha \boldsymbol{\theta}^{[t]}$, where $\alpha \in (0, 1]$. The impact of this regularization depends on the nature of the smoothed parameter: 
if applied to the \textit{mean}, it dampens the oscillations of the distribution's center induced by the stochasticity of the elite set; 
if applied to the \textit{covariance matrix}, it delays the premature collapse of the density (\textit{narrowing}), maintaining sufficient exploration of the search space. 
This mechanism ensures a robust trajectory toward the optimal point of the parameter space, balancing historical memory with new sample information.
These modifications are included in Algorithm \ref{algo:vf}, which represents a more efficient version of Algorithm \ref{algo:basic}. 
\begin{algorithm}
	\caption{CE Method for Maximization - Efficient vers.}
	\small
	\label{algo:vf}
	\begin{algorithmic}[1]
		\State \textbf{Input:} $\boldsymbol{\theta}^{[0]}, N^{[0]}, N_{max}, N_{elite}, \Delta N, $	$\alpha$ \textit{(Optional)}
		\State \textbf{Initialize:} 
		$t \gets 0$, 
		$N^{[t]} \gets N^{[0]}$, 
		$\mathbf{x}^\star \gets \widehat{\boldsymbol{\mu}}(\boldsymbol{\theta}_0)$, 
		$y^\star \gets \mathcal{J}(\bx^\star)$, 
		$y^{[t-1]} \gets y^\star$
		\Repeat
		\State \textbf{Generate:} $\mathcal{X}^{[t]} = \{\bX_{k} \sim q(\bx; \boldsymbol{\theta}^{[t]}) \}_{k=1}^{N^{[t]}} $
		\State \textbf{Evaluate:} $\mathcal{Y}^{[t]} = \{ \mathcal{J}(\bX_k) \}_{k=1}^{N^{[t]}}$
		\State \textbf{Sort:} $\mathcal{Y}^{[t]}$ in $Y_{(1)} \leq \dots \leq Y_{(N^{[t]})}$  \label{step:generate}
		\State \underline{\textbf{Level Update:}} $y^{[t]} \gets Y_{(k^{[t]})}$ $k^{[t]} = N^{[t]} - N_{elite} + 1$
		\State \textbf{Identify Elite Set:} $\mathcal{E}^{[t]} = \{ \bX_{(k^{[t]})}, \dots, \bX_{(N^{[t]})} \}$
		\State \textbf{Batch Expansion} \textit{(Conditional)}:
		\quad\quad\If{$y^{[t]} < y^{[t-1]}$ \textbf{and} $N^{[t]} < N_{max}$}
		\State \textbf{Generate:} $\Delta \mathcal{X} = \{ \bX_k \sim f(\bx); \boldsymbol{\theta}^{[t]}) \}_{k=N^{[t]}+1}^{N_t + \Delta N }$ 
		\State \textbf{Evaluate:} $\Delta \mathcal{Y} = \{ \mathcal{J}(\bX_k) \mid \bX_k \in \Delta \mathcal{X} \}$
	  \State \textbf{Merge:} $\mathcal{X}^{[t]} \gets \mathcal{X}^{[t]} \cup \Delta \mathcal{X}$ and $\mathcal{Y}^{[t]} \gets \mathcal{Y}^{[t]} \cup \Delta \mathcal{Y}$
		\State $N^{[t]} \gets N^{[t]} + \Delta N$; $\rho^{\star}=\frac{N_{elite}}{N^{[t]}}$;
		\State \textbf{goto}  step \ref{step:generate}
		\EndIf
		\State \textbf{Maxim-um/izer Update}\\
		\quad  \quad \textbf{if} $Y_{(N^{[t]})} > y^\star$ \textbf{then} $y^\star \gets Y_{(N^{[t]})}, \bx^\star \gets \bX_{(N^{[t]})}$		
		\State \underline{\textbf{Parameter Update}} \textit{(CE Solution)}:
		\State \textbf{Compute}:
		\begin{equation}\label{Eq:Update2}
		\E_{\scriptsize\bbtheta^{[t]}}[\bt(\bX)] =\frac{\sum_{k=1}^{K}\mathds{1}_{y}(\bx_{k})\bt(\bx_{k})}{\sum_{k=1}^{K}\mathds{1}_{y}(\bx_{k})}= \frac{1}{|\mathcal{E}^{[t]}|} \sum_{\bx \in \mathcal{E}^{[t]}} \bt(\bx)
		\end{equation}
		\State \textbf{Update}: $\boldsymbol{\theta}^{[t+1]}$ based on \eqref{Eq:Update2}
		\State \textbf{Smoothing} \textit{(Optional)}  $\boldsymbol{\theta}^{[t+1]} \gets (1-\alpha) \boldsymbol{\theta}^{[t+1]} + \alpha \boldsymbol{\theta}^{[t]}$
		\State \textbf{Reset:}  $y^{[t-1]} \gets y^{[t]}$,  $N^{[t]} \gets N^{[0]}$,  $t \gets t + 1$
		\Until{\textit{Convergence}}
		\State \textbf{Output:} $\bx^\star, y^\star, \rho^{\star}$
	\end{algorithmic}
\end{algorithm}	
\begin{Rmk}\label{Rem:DiscreteCE}
The parametric approach that led to Problem \eqref{Prob:Opt_theta} is motivated by the fact that the distribution $q$ to optimize is a continuous function. This is a direct consequence of the fact that Problem \eqref{Prob:Problem_CE} has a continuous optimization variable. If, instead, a discrete problem is to be solved, then the function $q$ could be taken as a probability mass function (PMF) and Problem \eqref{Prob:Opt_q} would reduce to the discrete problem of optimizing the vector of the PMF values. This is the CE-based approach that has been used in related studies that employed the CE method for the optimization of RIS-aided networks, e.g. \cite{JChen2026}. The numerical analysis in Sec. \ref{sec:numerical_results} will show that the continuous version of the CE method enjoys a significantly lower complexity than its discrete counterpart. 
\end{Rmk}
\paragraph{Convergence}	\label{par:stop_criteria}
Let us now formally prove the convergence of the proposed modified version of the CE method, as outlined in Algorithm \ref{algo:vf}.
We consider two instances of Algorithm \ref{algo:vf}: \emph{(i)} \textit{CE-$\boldsymbol{\mu}$} in which $\bbtheta=\boldsymbol{\mu}$, i.e. only the distribution mean is updated; \emph{(ii)} \textit{CE-$(\boldsymbol{\mu},\boldsymbol{\sigma})$}, in which $\bbtheta=[\boldsymbol{\mu}, \boldsymbol{\sigma}]$, i.e. both the distribution mean and  variance are updated.
\begin{Thm}[Convergence of Algorithm \ref{algo:vf}]
	\label{thm:convergence_CE}
	Assume:
	\begin{enumerate}
		\item[\textbf{(A1):}] \label{ass:domain}  $\mathcal{X}$ is a compact 
		metric space.
		
		\item[\textbf{(A2):}] \label{ass:objective}  $\mathcal{J}: 
		\mathcal{X} \to \mathbb{R}$ is measurable, bounded above, and attains its 
		maximum $y^\star$ on a non-empty set $\mathcal{X}^\star \subset \mathcal{X}$.
		
		\item[\textbf{(A3):}] \label{ass:family} 
		$f_{\boldsymbol{\theta}}$ belongs to the $m$-EF \eqref{def:NEF_form} with 
		$\boldsymbol{\theta}\in\Theta \subseteq \mathbb{R}^m$.
		
		\item[\textbf{(A4):}]  \label{ass:support}
		For every $\delta > 0$ and every $t$, $P\!\left(\mathcal{J}(\bX_k) > y^\star - \delta\right) > 0$, 
		where $\mathbf{X}_k \sim f(\,\cdot\,,\boldsymbol{\theta}^{[t]})$ is generated 
		by the sampling mechanism of Algorithm~\ref{algo:vf}.

		\item[\textbf{(A5):}] \label{ass:smoothing}  Smoothing is applied only to the parameter $\boldsymbol{\sigma}$, in the \textit{CE-$(\boldsymbol{\mu},\boldsymbol{\sigma})$} variant. In any case $\sigma^{[t]} > 0$ for all $t$ whenever $N_{elite} \geq 2$ 
		and elite samples are not all identical.
		
		\item[\textbf{(A6):}] \label{ass:feasibility}  Every 
		generated sample $\mathbf{X}_k \in \mathcal{X}$ by construction of the 
		sampling mechanism.
	\end{enumerate}
	
	\noindent Then, as $t \to \infty$, Algorithm~\ref{algo:vf} satisfies almost 
	surely (a.s.):
	\begin{enumerate}
		\item[\textbf{(C1)}] \label{conv:level} $y^{[t]}$ converges to $y^\star$ from below, i.e. 
		$y^{[t]} \nearrow y^\star$;
		\item[\textbf{(C2)}] \label{conv:param}  $\boldsymbol{\theta}^{[t]}\!\to \!\boldsymbol{\theta}^\star$, where $\boldsymbol{\theta}^{\star}$ 
	fulfills \eqref{Eq:ProbStochastic} with $y\!=\!y^{\star}$ and $\rho=\rho^{\star}$.		
		\item[\textbf{(C3)}] \label{conv:optimizer}  $\mathcal{J}(\bx^\star) 
		\to y^\star$.
	\end{enumerate}	
	\begin{proof}
		\textbf{(C1).} 
		The proof generalizes stochastic approximation arguments from  
		\cite{homem_et_rubinstain_2002estimation}, \cite{shapiro1996simulation} to the case of Algorithm \ref{algo:vf}. To elaborate, 
		the elite threshold $y^{[t]} = Y_{(k^{[t]})}$ with 
		$k^{[t]} = N^{[t]} - N_{elite} + 1 \in \{1,\dots,N^{[t]}\}$ is a well-defined order 
		statistic by construction of Algorithm~\ref{algo:vf}. For any $\delta > 0$, 
		\textbf{(A4)} gives $P\!\left(\mathcal{J}(\bX_k) > y^\star - \delta\right) 
		=: p_\delta> 0$. Among $N^{[t]}$ draws, at least one satisfies 
		$\mathcal{J}(\bX_k) > y^\star - \delta$ with probability 
		$1-(1-p_\delta)^{N^{[t]}} \to 1$, so $Y_{(N^{[t]})} \to y^\star$ a.s. Then, the
		monotonicity of $\{y^{[t]}\}$ ensures that $y^{[t]} \nearrow y^\star$ a.s.
		
		\smallskip
		\textbf{(C2).}
		By construction, Problem \eqref{Prob:Opt_theta} is globally solved in each iteration. By the 
		strong law of large numbers, the compactness of the feasible set \textbf{(A1)}, and standard 
		M-estimation theory \cite{rubinstein_et_shapiro_1993discrete}, it follows that 
		$\frac{1}{|\mathcal{E}|}\sum_{\scriptsize{\bx}_{k}\in\mathcal{E}}\bt(\bx_{k})$ tends to	$\E[\bt(\bX)]$
		 uniformly a.s., so 
		$\boldsymbol{\theta}^{[t]} \to \boldsymbol{\theta}^\star$ a.s. by 
		identifiability of the $m$-EF \cite{rubinstein_et_shapiro_1993discrete}.
		As for CE-$\boldsymbol{\mu}$, the variance $\boldsymbol{\sigma}$ is fixed and only 
		$\boldsymbol{\mu}{^{[t]}} \to \boldsymbol{\mu}^\star$ is updated, with $\alpha = 0$; then, the 
		convexity of the set $\Theta$ \textbf{(A3)} is preserved trivially.
		As for CE-$(\boldsymbol{\mu},\boldsymbol{\sigma})$, smoothing on $\boldsymbol{\sigma}$ 
		keeps $\boldsymbol{\sigma}^{[t]} > 0$ and $\boldsymbol{\theta}^{[t+1]} \in \Theta$ for any $t$, and convergence follows by stochastic approximation 
		arguments from \cite{homem_et_rubinstain_2002estimation,shapiro1996simulation}. Then, upon convergence, \eqref{Eq:ProbStochastic} is fulfilled with $y=y^{\star}$ and $\rho=\rho^{\star}$. 
		
		\smallskip
		\textbf{(C3).} By \textbf{(A6)}, $\mathbf{x}^\star \in \mathcal{X}$ always. 
		Since $\mathcal{J}(\mathbf{x}^\star) = \max_{s \leq t} Y_{(N_s)}$ is 
		non-decreasing, bounded from above by $y^\star$, and its limit equals $y^\star$ by 
		\textbf{(C1)}, the conclusion follows.
		
		\smallskip
		\textit{Batch expansion.} The conditional expansion (Algorithm~\ref{algo:vf}, 
		steps 9-15) draws additional i.i.d. feasible samples from $f_{\boldsymbol{\theta}^{[t]}}$. The merged set still contains i.i.d. samples\ from 
		$f_{\boldsymbol{\theta}^{[t]}}$, so all arguments above apply without modification.
	\end{proof}
\end{Thm}
\noindent Proposition \ref{thm:convergence_CE} ensures that Algorithm \ref{algo:vf} generates, a. s., a non-decreasing sequence of objective values, and converges both in the objective value $\{y^{[t]}\}_{t}$ and parameter vector $\{\boldsymbol{\theta}^{[t]}\}_{t}$. 	

\paragraph{Computational complexity} In each iteration of Algorithm \ref{algo:vf}, the bulk of the complexity is related to the identification of the elite set to compute \eqref{Eq:Update2}. Denoting by $V$ the number of optimization variables, in iteration $t$ it is necessary to sample $V$ feasible vectors of $N^{[t]}$ components each. The $N_{elite}$ samples of each of the $V$ vectors can be determined by the min-heap sort method, which has complexity $N^{[t]}\log(N_{elite})$ \cite{MinHeapSort}. Thus, the complexity of each iteration of Algorithm \ref{algo:vf} scales with $VN^{[t]}\log(N_{elite})$, which leads to the following asymptotic complexity of Algorithm \ref{algo:vf} 
\begin{align}\label{Eq:ComplexityCE}
\mathcal{C}_{\text{CE}}&=\mathcal{O}\left(\sum_{t=1}^{T}VN^{[t]}\log (N_{elite})\right)\notag\\
&\leq\mathcal{O}\left(TVN_{max}\log (N_{elite})\right)\;,
\end{align}
with $T$ the number of iterations to reach convergence. Thus, it can be seen that Algorithm \ref{algo:vf} enjoys a linear complexity in the number of optimization variables $V$. 

\subsection{Metropolis-Hastings maximization}
The MH algorithm \cite{hastings1970monte,metropolis1953equation} is a method for sampling probability distributions that are too complex to be sampled directly. It was originally conceived in statistical mechanics to estimate macroscopic quantities in high-dimensional systems. In this framework, the state of the system is modeled as a random vector $\bX$ whose target probability density $\pi(\bx)$ is defined \textit{a priori} according to the Boltzmann distribution:
\begin{equation}\label{Eq:MH_Stationary}
	\pi(\bx) = \frac{1}{Z} e^{-\beta U(\bx)}
\end{equation}
where $U(\bx)$ represents the \textit{potential energy} of the configuration, $Z$ is the partition function, and $\beta > 0$ is related to the temperature of the system. A fundamental property of this distribution is its intrinsic tendency to concentrate probability mass on states with lower energy, effectively centering the measure around the global minima of $U(\bx)$. This allows to make a connection with optimization problems. To elaborate, let us consider Problem \eqref{Prob:Problem_CE} and let us take $U(\bx)=-\mathcal{J}(\bx)$. Then, the samples from \eqref{Eq:MH_Stationary}, being with high probability close to the global minimum of $U(\bx)$, will also be close to the global maximum of $\mathcal{J}(\bx)$ with high probability. However, directly sampling from $\pi(\bx)$ is generally intractable due to the high dimensionality of the space $\mathcal{X}$, which motivates the use of the MH method to obtain samples from $\pi(\bx)$. 

The MH algorithm works by designing a Markov Chain Monte Carlo (MCMC) process by constructing transition probabilities, and generating accordingly a sequence of states $\{\bx_t\}$, such that the chain admits $\pi(\bx)$ as its unique stationary distribution. As for the existence of a stationary distribution, a sufficient condition is the so-called \emph{Detailed Balance} condition, which states that
\begin{equation}\label{Eq:DB_cond}
P(\bx|\by)\pi(\by)=P(\by|\bx)\pi(\bx)\;,
\end{equation}
wherein $P(\bx|\by)$ denotes the probability to transition from state $\by$ to state $\bx$. In order to fulfill this condition, the MH algorithm defines the transition probabilities as $P(\bx|\by)=q(\bx|\by)a(\bx,\by)$, with $q(\cdot|\cdot)$ being the \textit{proposal distribution} and $a(\cdot,\cdot)$ the probability of accepting the new state $\bx$, when the current state is $\by$. Then, \eqref{Eq:DB_cond} can be rewritten as 
\begin{equation}
\frac{a(\bx,\by)}{a(\by,\bx)}=\frac{q(\by|\bx)}{q(\bx|\by)}\frac{\pi(\bx)}{\pi(\by)}\;,
\end{equation}
which can be seen to be fulfilled, for any $q(\cdot|\cdot)$, by choosing 
\begin{equation}\label{Eq:stationaryDB}
a(\bx,\by)=\min\left\{1,\frac{q(\by|\bx)}{q(\bx|\by)}\frac{\pi(\bx)}{\pi(\by)}\right\}=\min\{1,\alpha\}\;.
\end{equation}
 The quantity $\alpha$ is called the \textit{Hastings ratio}, and, plugging \eqref{Eq:MH_Stationary} into \eqref{Eq:stationaryDB}, can be written as 
 \begin{equation}
\alpha=\frac{q(\by|\bx)}{q(\bx|\by)}\ds e^{\beta \Delta {\mathcal J}}\;,
 \end{equation}
 wherein $\Delta{\mathcal J}= {\mathcal J}(\bx)- {\mathcal J}(\by)$ gives the variation of the objective function as we move from $\by$ to $\bx$. The described mechanism can be implemented following a \textit{two-step} stochastic rule referred respectively as \textit{proposal generation} phase and \textit{acceptance/rejection} phase:
\begin{enumerate}
	\item\textbf{\textit{Proposal:}} Given the state $\bx^{[t]}$ at iteration $t$, a new candidate state $\bx$ is sampled from a proposal distribution $q(\bx | \bx_t)$.
	\item \textbf{\textit{Acceptance/Rejection:} }The transition is accepted with probability $\rho = \min(1, \alpha)$.
\end{enumerate}
We observe that, in the case of symmetric proposal distribution (i.e. when $q(\bx|\by) = q(\by|\bx)$), the ratio simplifies to $\pi(\bx)/\pi(\by)$, which is the criterion proposed in the first instance of the MH method \cite{metropolis1953equation}.

As for the uniqueness of the stationary distribution, it is necessary that the generated Markov chain is ergodic. To elaborate, we first recall the following definitions from \cite{MeynTweedie2009}.
\begin{Def}[$\pi$-\textit{Irreducibility}]
	The chain is said to be $\pi$-irreducible if, for every set $A$ with $\pi(A) > 0$, the probability of reaching $A$ starting from any $\mathbf{x} \in \mathcal{X}$ is \textit{strictly positive} in a \textit{finite number} of steps:
\end{Def}
\begin{Def}[\textit{Aperiodicity}]
	An irreducible chain is \textit{aperiodic} if there is no partition of the state space into disjoint sets $\mathcal{X} = D_1 \cup \dots \cup D_d$ (with $d \ge 2$) such that the chain moves deterministically in a cycle $D_1 \to D_2 \to \dots \to D_1$.  \\
	In MH schemes, a non-zero rejection probability is a \textit{sufficient condition} for aperiodicity.
\end{Def}
\begin{Def}[\textit{Positive Harris Recurrence}]
	A chain is Harris recurrent if it visits any set $A$ with $\pi(A) > 0$ infinitely often with probability $1$ for all starting states $\bx \in \mathcal{X}$. The chain is \textit{positive} if it admits a unique invariant probability measure $\pi$, which implies finite expected return times to accessible sets.
\end{Def}
\paragraph{Convergence} Based on the definitions above, we have the following result, which ensures the convergence of the sequence of accepted states generated by the MH method to a sample from the unique stationary distribution, from any initial state and with any proposal distribution.
\begin{proposition}
	\label{thm:ergodic_theorem}
	If a Markov chain is $\pi$-irreducible, aperiodic, and positive Harris recurrent, then it admits $\pi$ as its unique stationary distribution. For every starting state $\mathbf{x}_0 \in \mathcal{X}$, the distribution of the chain converges in total variation distance:
	\begin{equation}
		\lim_{t \to \infty} \| P^{[t]}(\bx_0, \cdot) - \pi(\cdot) \|_{TV} = 0
	\end{equation}
\end{proposition}
\noindent The general MH procedure is reported in Algorithm \ref{algo:MH_final}, whose convergence is guaranteed by Proposition \ref{thm:ergodic_theorem}.
\begin{algorithm}
\small
	\caption{Metropolis-Hastings (Loss-Adaptive)}
	\label{algo:MH_final}
	\begin{algorithmic}[1]
		\State \textbf{Input:} $\mathcal{J}(\cdot)$, $q(\cdot|\cdot)$, $\beta$, $T$ (max iterations)
		\State \textbf{Initialize:} $t=0, \mathbf{x}^{[t]}, \mathbf{x}^\star \gets \mathbf{x}^{[t]}, y^{[t]} = \mathcal{J}(\mathbf{x}^{[t]}), y^\star \gets y^{[t]}$
		\Repeat
		\State \textbf{Sample} $\mathbf{x} \sim q(\cdot | \mathbf{x}^{[t]})$ and \textbf{Compute} $F = \mathcal{J}(\mathbf{x})$
		\State \textbf{if} $F > y^\star$ \textbf{then} $(\mathbf{x}^\star, y^\star) \gets (\mathbf{x}, F)$ \hfill (Global update)
		\State $\Delta \mathcal{J} \gets F - y^{[t]}$
		\State $\alpha \gets \exp(\beta \Delta \mathcal{J}) \cdot \frac{q(\mathbf{x}^{[t]} | \mathbf{x})}{q(\mathbf{x} | \mathbf{x}^{[t]})}$ \Comment{Hastings ratio}
		\State $\rho \gets \min(1, \alpha)$
		\State \textbf{Sample} $u \sim \mathcal{U}(0,1)$
		\If{$u < \rho$} $(\mathbf{x}^{[t+1]}, y^{[t+1]}) \gets (\mathbf{x}, F)$ \hfill (Accept)
		\Else $\quad (\mathbf{x}^{[t+1]}, y^{[t+1]}) \gets (\mathbf{x}^{[t]}, y^{[t]})$ \hfill (Reject)
		\EndIf
		\State $t \gets t + 1$
		\Until {Convergence} or $t = T$
		\State \textbf{Output:} $\mathbf{x}^\star$, $y^\star$
	\end{algorithmic}
\end{algorithm}
\begin{Rmk}[Role of the parameter $\beta$]
	The inverse temperature $\beta$ governs the trade-off between \textbf{exploration} and \textbf{exploitation}. In the limit $\beta \to 0$, the algorithm reduces to a pure random walk. Conversely, as $\beta \to \infty$, the sampler becomes increasingly greedy, accepting only moves that strictly improve the objective. Tuning $\beta$ allows escaping local maxima while ensuring concentration around global peaks.
\end{Rmk}
\begin{Rmk}
The optimality properties of the MH method are weaker than those of the CE method, since the MH method does not guarantee to yield, upon convergence, a point that provides guaranteed level of the objective value with high probability. On the other hand, it exhibits a lower computational complexity than the CE method, as discussed below, and supported by the numerical analysis in Sec. \ref{sec:numerical_results}.
\end{Rmk}
\paragraph{Complexity} Unlike the CE method, Algorithm \ref{algo:MH_final} draws only one sample $\bx$ from the feasible set in each iteration, and therefore the complexity of each iteration scales linearly with the length $\bx$, i.e. the number of optimization variables $V$. Thus, if $T$ iterations are required to reach convergence, the overall asymptotic complexity of the MH method is given by 
\begin{equation}\label{Eq:ComplexityMH}
{\mathcal C}_{MH}=\mathcal{O}\left(TV\right)\;,
\end{equation} 
which is linear in $V$, as the complexity of the CE method given in \eqref{Eq:ComplexityCE}, but with a smaller coefficient that multiplies $V$.

\section{Application to RIS-aided Wireless Systems}\label{sec:case_study}\label{s:case_study}
The considered framework is extremely general, requiring only very mild assumptions on objective and constraint functions, and, thus, can be applied to a wide range of scenarios. As examples, this section shows how both the CE and MH methods can be applied to two instances of RIS-aided wireless communication networks. Specifically, both scenarios consider the uplink of a wireless network in which $K$ single-antenna mobile terminals communicate with a base station (BS) equipped with $N_{R}$ antennas. The mobile users reach the BS through an RIS equipped with $N_{RIS}$ elements. In both cases, the goal is to optimize the system EE with respect to the RIS reflection coefficients. However, the type of metasurface that is employed is different in each case.
\begin{itemize}
\item\textbf{Case 1:} The RIS is \emph{nearly-passive}, with the $n$-th element expressed as $\gamma_{n}=e^{j\phi_{n}}$, $\phi_{n}\in[0,2\pi]$, for all $n=1,\ldots,N$. In this case, let us define $\bphi=[\phi_{1},\ldots,\phi_{N}]$.
\item\textbf{Case 2:} The RIS is \emph{active} through the use of an analog amplifier. In this case, $\gamma_{n}$ is a complex number whose modulus and phase can both be optimized. 
\end{itemize}
In both cases, we define the $N_{RIS}\times N_{RIS}$ matrix as $\bGamma=\text{diag}(\gamma_{1},\ldots,\gamma_{N_{RIS}})$. The two scenarios are further elaborated and treated separately in the next two subsections. 
\subsection{Case 1. Nearly-passive RIS}\label{subsec:case1}
Let us denote by $p_{k}$, $\bh_{k}$, and $\bc_{k}$, the $k$-th user's transmit power, $N_{RIS}\times 1$ channel to the RIS, and $N_{R}\times 1$ receive filter, while $\bG$ and $\sigma_{BS}^{2}$ denote the $N_{R}\times N$ channel between the RIS and the BS, and the thermal noise power at the BS. Since the RIS is nearly-passive, it does not consume any radio-frequency power, but only static power required to configure the phase shifts. Thus, there is no noise amplification at the RIS, and the achievable rate enjoyed by user $k$ is
\begin{equation}
R_{k}=\log_{2}\left(1+\frac{p_{k}|\bc_{k}^{H}\bG\bGamma\bh_{k}|^{2}}{\sigma_{BS}^{2}\|\bc_{k}\|^{2}+\sum_{\ell\neq k}p_{\ell}|\bc_{k}^{H}\bG\bGamma\bh_{\ell}|^{2}}\right)\;, 
\end{equation}
which, upon optimal linear minimum mean squared error (LMMSE) reception, becomes
\begin{equation}
R_{k}=\log_{2}\!\left(\!1\!+\!p_{k}\bv_{k}^{H}\!\left(\!\sum_{\ell\neq k}p_{\ell}\
\bv_{\ell}\bv_{\ell}^{H}\!+\!\sigma_{BS}^{2}\bI_{N_{R}}\!\right)^{-1}\!\!\!\bv_{k}\right), 
\end{equation}
with $\bv_{\ell}=\bG\bGamma\bh_{\ell}$ for all $\ell=1,\ldots,K$. 
Therefore, the sum-rate maximization problem is formulated as 
\begin{subequations}\label{Prob:Case1_EE}
\begin{align}
&\ds\max_{\bphi}\sum_{k=1}^{K}R_{k}(\bphi)\label{Prob:Case1_EEa}\\
&\;\text{s.t.}\;\phi_{n}\in[0,2\pi]\;,\forall\; n=1,\ldots,N\label{Prob:Case1_EEb}
\end{align}
\end{subequations}
Despite its apparent simplicity, Problem \eqref{Prob:Case1_EE} is a non-convex problem for which no known method exists to obtain the optimal solution with polynomial complexity. Indeed, all approaches proposed in the literature are based on sub-optimal, iterative approaches, among which, the simplest one is the AO, which alternatively  optimizes one variable $\phi_{n}$ at a time, until convergence is reached. However, this approach still requires solving $NI$ scalar problems wherein $I$ is the number of iterations over the complete set of $N$ variables that are required to reach convergence. Moreover, each scalar problem is a non-convex problem which requires a line search in the set $[0,2\pi]$ or suboptimal gradient-based techniques, which further increase the complexity. Other widely-used approaches in the literature are based on sequential optimization methods and/or semi-definite relaxation approaches, which, however, have higher, although polynomial, complexity than AO. The rest of this section shows how to tackle \eqref{Prob:Case1_EE} by the CE and MH methods. 

\subsubsection{CE method for \eqref{Prob:Case1_EE}}\label{Sec:CE_1}
In order to apply the CE method, we set $\bx=\bphi$, and generate samples $\bphi\sim q(\cdot,\bbtheta)$. As sampling distribution, we consider the Von Mises PDF, for which it holds $\bt(\bphi)=[\cos(\bphi),\; \sin(\bphi)]$, and which is the typical choice the $m$-EF family to account for the periodic variables, like the periodicity of \eqref{Prob:Case1_EEa} in $\bphi$. Moreover, we consider both the CE version that updates mean and variance, i.e. CE-($\boldsymbol{\mu}, \boldsymbol{\sigma}$), and that which updates only the mean, i.e, CE-$\boldsymbol{\mu}$. As for the map between $\mathbb{E}[\bt(\bphi)]$ and the parameters of the Von Mises distribution, it requires the solution of a non-linear equation involving Bessel functions. Aiming at lower complexity approach, it is also possible to approximate the Von Mises PDF by a wrapped Gaussian PDF, recalling that the wrapped Gaussian tends to the Von Mises when the variance is small compared to $\pi$ \cite{jammalamadaka_and_sengupta_2001_topics}. This is the case of interest, because: (1) as for CE-$\boldsymbol{\mu}$, the variance needs not be updated, and thus can be set to any sufficiently small value so that the wrapped Gaussian approximation holds, while ensuring the generation of a statistically significant data sample in each iteration; (2) as for CE-($\boldsymbol{\mu}, \boldsymbol{\sigma}$) the  variance can be initialized to a sufficiently small value, which is then further reduced during the execution of the CE algorithm\footnote{By construction, the CE method iteratively updates mean and variance in order to center the mean over a maximizer of the objective function and progressively reduce the variance.}. Embracing the wrapped Gaussian approximation, the connection between $\mathbb{E}[\bt(\bphi)]$ and the mean and variance of the distribution is\footnote{Note that the same formula for the update of $\boldsymbol{\mu}$ holds for the Von Mises distribution, too.}  \cite{jammalamadaka_and_sengupta_2001_topics}:
\begin{equation}
\boldsymbol{\mu}=\angle{\left\{\sum_{\small\bphi_{k} \in \mathcal{E}_t} \frac{\bt(\bphi_{k})}{|\mathcal{E}_t|}\right\}}\;,\boldsymbol{\sigma}=\sqrt{-2\ln\left|\sum_{\small\bphi_{k} \in \mathcal{E}_t} \frac{\bt(\bphi_{k})}{|\mathcal{E}_t|}\right|}\;.
\end{equation}
\begin{proposition}\label{Prop:CE_Case1}
Given the formulation above, applying Algorithm \ref{algo:vf}, with either variant CE-$\boldsymbol{\mu}$ or CE-($\boldsymbol{\mu}, \boldsymbol{\sigma})$, converges and fulfills properties \textbf{(C1)}, \textbf{(C2)}, and \textbf{(C3)} of Proposition \ref{thm:convergence_CE}.
\end{proposition}
\begin{proof}
In order to show the results, we verify that Assumptions \textbf{(A1)}-\textbf{(A6)} are fulfilled.  
\textbf{(A1)} holds because $[0,2\pi]^N_{RIS}$ is compact. \textbf{(A2)} holds because$\sum_{k=1}^{K}R_{k}(\bphi)$ is continuous on $[0,2\pi]^N_{RIS}$, 
	hence bounded and attains its maximum; 
\textbf{(A3)} holds because the Von Mises belongs to the NEF. 
\textbf{(A4)} holds because the because the Von Mises PDF assigns a positive probability density to every open
arc in $[0,2\pi)$. This is also true for the CE-($\boldsymbol{\mu}, \boldsymbol{\sigma}$) variant, because, in each iteration $t$, it holds $\sigma^{[t]} > 0$. \textbf{(A5)} holds for both CE-$\boldsymbol{\mu}$, because $\bar{\boldsymbol{\sigma}}$ is fixed and positive throughout, and CE-($\boldsymbol{\mu}, \boldsymbol{\sigma})$, because $\sigma^{[t]} > 0$ for all $t$ whenever $N_{elite} \geq 2$ and elite samples are not all identical; smoothing on $\boldsymbol{\sigma}$ prevents premature degeneracy. \textbf{(A6)} holds because all samples lie in $[0,2\pi]^N_{RIS}$ by construction.
\end{proof}

\subsubsection{MH method} \label{Sec:MH_1}
In order to apply the MH method, we set $\bx=\bphi$ and generate candidate samples from the distribution $q(\bphi|\boldsymbol{\mu})$, with $q$ being a wrapped multi-variate Gaussian over the interval $[0,2\pi]^{N_{RIS}}$, $\boldsymbol{\mu}$ the mean vector, and covariance matrix $\bSigma=\text{diag}(\sigma_{1}^{2},\ldots\sigma_{N_{RIS}}^{2})$. With this formulation, we can directly apply Algorithm \ref{algo:MH_final} to solve Problem \eqref{Prob:Case1_EE} and the following Proposition holds. 
\begin{proposition}\label{Prop:MH_Case1}
Given the formulation above, applying Algorithm \ref{algo:MH_final} generates a Markov chain that converges in total variation distance to its unique stationary distribution. 
\end{proposition}
\begin{proof}
	The result follows upon verifying the conditions of Proposition \ref{thm:ergodic_theorem}. To show this, we observe that the considered wrapped Gaussian density  $q(\bphi|\boldsymbol{\mu})$ is strictly positive over $[0, 2\pi)^N$ and over the entire domain $\mathcal{X}$. The chain is therefore \textit{strongly irreducible} (with parameter $n=1$). Moreover, the fact that the MH method has a non-zero probability of rejecting a candidate sample ensures \textit{aperiodicity}. Combined with the compactness of the feasible sets, these properties guarantee \textit{positive Harris recurrence} and, thus, unique convergence to the unique stationary distribution $\pi(\mathbf{x})$.
\end{proof}

\subsection{Case 2. Active RIS}\label{subsec:case2}
Let us now extend the previous approach to the case in which the RIS is equipped with analog amplifiers. In this case, due to the presence of the transmit amplifiers, the RIS consumes not only static power, but also radio-frequency power, and introduces additional thermal noise. Moreover, we deploy the RIS in the near-field of the BS array, as a reconfigurable holographic beamforming structure. This allows reducing the number of transmit antennas to $N_{R}=1$, while still ensuring large beamforming and EE gains  \cite{Fotock_RHB}. In this case, the achievable, rate of user $k$ can be shown to be
\begin{equation}\label{Eq:Rate_Case2}
	R_{k}\!=\!\log_{2}\!\left(\!1\!+\!\frac{p_{k}|\bg^{H}\bGamma\bh_{k}|^{2}}{\sigma_{BS}^{2}\!+\!\sigma_{RIS}^{2}\bg^{H}\bGamma\bGamma^{H}\bg\!+\!\sum_{\ell\neq k}p_{\ell}|\bg^{H}\bGamma\bh_{\ell}|^{2}}\right),
\end{equation}
wherein $\bg$ is the $N_{RIS}\times 1$ channel vector from the RIS to the BS, while $\sigma_{BS}^{2}$ and $\sigma_{RIS}^{2}$ are the thermal noise power at the BS and RIS, respectively. As for the power consumption, the input and output powers at the RIS are expressed as \cite{Fotock_RHB}
\begin{align}
	P_{in}&=\textstyle\sum_{k=1}^{K}p_{k}\|\bh_{k}\|^{2}+\sigma_{RIS}^{2}\\
	P_{out}&=\textstyle\sum_{k=1}^{K}p_{k}\|\bGamma\bh_{k}\|^{2}+\sigma_{RIS}^{2}\tr(\bGamma\bGamma^{H})
\end{align}
Thus, the radio-frequency power consumed by the RIS is the difference between the output and the input power, i.e. $P_{out}-P_{in}$, which is non-negative since the RIS is active, but can not be larger than the maximum power that the power amplifier can provide $P_{r,max}$. Denoting by $P_{c,n}$ the static power consumption of each RIS element, and by $P_{0}$ the static power consumption of all other nodes in the network, the total system power consumption is written as $P_{T}=P_{out}-P_{in}+\sum_{k=1}^{K}p_{k}+N_{RIS}P_{c,n}+P_{0}=\sum_{k=1}^{K}p_{k}\|\bGamma\bh_{k}\|^{2}+\sigma_{RIS}^{2}\tr(\bGamma\bGamma^{H})+P_{a}$, wherein $P_{a}=\sum_{k=1}^{K}p_{k}+N_{RIS}P_{c,n}+P_{0}-P_{in}$. Finally, the EE maximization problem can be stated as 
\begin{subequations}\label{Prob:Case2_EE}
	\begin{align}
		&\ds\max_{\bGamma}\;\frac{\sum_{k=1}^{K}R_{k}(\bGamma)}{\sum_{k=1}^{K}p_{k}\|\bGamma\bh_{k}\|^{2}+\sigma_{RIS}^{2}\tr(\bGamma\bGamma^{H})+P_{a}}\\
		&\;\text{s.t.}\;0\leq\textstyle\sum_{k=1}^{K}\!p_{k}\|\bGamma\bh_{k}\|^{2}\!\!+\!\sigma_{RIS}^{2}\tr(\bGamma\bGamma^{H})\!-\!P_{in}\leq P_{r,max}\label{Eq:PowerConstr}
	\end{align}
\end{subequations}
wherein the constraint enforces that the radio-frequency power consumed by the RIS can not exceed the power provided by the power amplifiers, and that the RIS is operating in the active mode, i.e. $P_{out}\geq P_{in}$. Problem \eqref{Prob:Case2_EE} is a challenging non-convex, fractional program, for which no known method exists to find the global solution with polynomial complexity. Traditional approaches in the literature resort to convex relaxations, e.g. the sequential optimization, which admit polynomial complexity. However, given the large number of optimization variables, polynomial complexity may not be practical. Instead, the CE and MH methods can tackle Problem \eqref{Prob:Case2_EE}, with linear complexity in the number of RIS elements. 

\subsubsection{CE method for \eqref{Prob:Case2_EE}}
In this case, we need to optimize both the phases and moduli of the RIS reflection coefficients. To begin with, let us set $\bx=\bgamma=[\gamma_{1},\ldots,\gamma_{N_{RIS}}]$. Then, the constraint in \eqref{Eq:PowerConstr} can be written as 
\begin{align}
P_{in}&\!\leq \!\sum_{n=1}^{N_{RIS}}\!|\gamma_{n}|^{2}\sum_{k=1}^{K}\left(p_{k}|h_{k,n}|^{2}\!+\!\sigma_{RIS}^{2}\right)\!\leq \!P_{r,max}\!+\!P_{in}\;,\notag
\end{align}
which implies that the constraint is coupled in the moduli of the RIS coefficients. In other words, the upper and lower limits in which each RIS coefficient can vary, depends on the values of the other coefficients. Indeed, defining $I_{i}=\sum_{n\neq i}^{N_{RIS}}\!|\gamma_{n}|^{2}\sum_{k=1}^{K}\left(p_{k}|h_{k,n}|^{2}\!+\!\sigma_{RIS}^{2}\right)$,
for any $i=1,\ldots,N$, \eqref{Eq:PowerConstr} can be written as $b_{i}^{-}\leq |\gamma_{i}|\leq b_{i}^{+}$, with 
\begin{align}\label{Eq:CoupledModuli}
b_{i}^{-}&\!=\!\sqrt{\frac{P_{in}-I_{i}}{\ds\sum_{k=1}^{K}p_{k}|h_{k,n}|^{2}\!+\!\sigma_{RIS}^{2}}}\;,\;{b_{i}^{+}=\sqrt\frac{P_{r,max} + P_{in}-I_{i}}{\ds\sum_{k=1}^{K}p_{k}|h_{k,n}|^{2}\!+\!\sigma_{RIS}^{2}}}
\end{align}
As a consequence, while the optimization of the phases of the RIS coefficients can be treated through a multi-variate wrapped Gaussian, as shown in Sec. \ref{Sec:CE_1}, a different approach is required for the optimization of the moduli. Specifically, while the data samples of the phases are generated in parallel for all $i=1,\ldots,N$, the generation of the samples of the moduli must account for the fact that the feasible set for one modulus, say $|\gamma_{i}|$, depends on the other moduli $\{|\gamma_{n}|\}_{n\neq i}$, due to \eqref{Eq:CoupledModuli}. Thus, the samples of the modulus of each RIS coefficient are generated sequentially for $i=1,\ldots,N_{RIS}$, as a truncated Gaussian distribution \cite{jammalamadaka_and_sengupta_2001_topics}, i.e. $|\gamma_{i}|\sim q(\varrho_{i},\bbtheta)$, with 
\begin{equation}
q(\varrho_{i},\bbtheta)=c_i(\bgamma) \cdot \mathcal{N}(\varrho_{i},\mu_{i},\sigma_{i}^{2}) \mathbb{I}_{\mathcal{B}}(\varrho_i)\;,
\end{equation}
wherein $\mathcal{N}(\varrho_{i},\mu_{i},\sigma^{2})$ denotes the Gaussian PDF with mean $\mu_{i}$ and variance $\sigma_{i}^{2}$, $\mathbb{I}_{\mathcal{B}}(\varrho_i)$ is the indicator function of the interval $\mathcal{B}=[b_{i}^{-},b_{i}^{+}]$, making sure that all generated samples are feasible, and $c_i(\bgamma)$ is the normalization term ensuring the PDF integrates to one, which is defined as 
\begin{equation}
	c_i(\bgamma) = \left[ \Phi\left( \frac{b_i^+ - \varrho_i}{\sigma_{i}} \right) - \Phi\left( \frac{b_i^{-} - \varrho_i}{\sigma_{i}} \right) \right]^{-1}\;,
\end{equation}
with $\Phi(\cdot)$ the standard Gaussian cumulative distribution function. Since the Gaussian distribution belongs to the NEF, also its truncated version does, as discussed in Remark \ref{Rmk:WrappedTruncated}. For this use-case, we consider the CE-$\boldsymbol{\mu}$ variant of the algorithm, which updates only the mean value\footnote{The extension to CE-$(\boldsymbol{\mu},\boldsymbol{\sigma})$ is straightforward and is omitted for brevity.}, i.e. $\bbtheta=\boldsymbol{\mu}$. Finally, as far as the optimization of the moduli is concerned, we set $t_{i}(\varrho_{i})=\varrho_{i}$, for $i=1,\ldots,N$. Moreover, we observe that the constraint in \eqref{Eq:PowerConstr} is not jointly convex with respect to the moduli $|\gamma_{1}|,\ldots,|\gamma_{N}|$, due to the lower bound inequality. This implies that, when updating the mean value of the moduli variables, the updated mean can lie outside the feasible set, which may lead to subsequent generations of unfeasible data points. In order to avoid this, we project the computed mean onto the feasible set, by simply applying a scaling factor to the vector of the moduli. With this formulation, we can apply Algorithm \ref{algo:vf} to solve Problem \eqref{Prob:Case1_EE}, and the following proposition holds. 
\begin{proposition}\label{Prop:CE_Case2}
Given the formulation above, Algorithm 2 converges and fulfills properties \textbf{(C1)}-\textbf{(C3)} of Proposition \ref{thm:convergence_CE}.
\end{proposition}
\begin{proof}
\textbf{(A1)} holds because the feasible set of \eqref{Prob:Case2_EE} is compact. \textbf{(A2)-(A3)} hold because the amplitudes are drawn from a truncated 
distribution that belongs to the NEF, with full support on each feasible interval. \textbf{(A4)} holds because, by construction and due to the continuity of the objective function, the update mechanism ensures that in each iteration $t$, the probability that the objective is larger than $y^\star - \delta$ is strictly positive for every $\delta > 0$. \textbf{(A5)} holds because the variance $\sigma$ is initialized to a positive value and is not updated throughout the execution of the algorithm. \textbf{(A6)} holds because the update mechanism ensures feasibility at each draw by construction.
\end{proof}

\subsubsection{MH method for \eqref{Prob:Case2_EE}}
In order to apply the MH method, we generate the samples of the moduli of the reflection coefficients again from the truncated Gaussian distribution $q(\varrho_{i} | \bgamma^{[t]}) = c_i(\bgamma^{[t]}) \!\cdot \!\mathcal{N}(\varrho_i,\mu_{i}, \sigma^2) \mathbb{I}_{\mathcal{B}^{[t]}}(\varrho_i)$, with $\!\bgamma^{[t]}\!$ the vector $\!\bgamma\!$ in the $t$-th iteration. As for the optimization of the RIS phases, the same approach as in Sec. \ref{Sec:MH_1} is used. The following result holds. 
\begin{proposition}\label{Prop:MH_Case2}
Given the formulation above, applying Algorithm \ref{algo:MH_final} generates a Markov chain that converges in total variation distance to its unique stationary distribution. 
\end{proposition}
\begin{proof}
	The result follows similarly as for Proposition \ref{Prop:MH_Case1}.
\end{proof}

\section{Numerical Results}\label{sec:numerical_results}
This section numerically analyzes the performance and running time of the CE and MH methods, comparing them to several benchmark algorithms. 

\subsection{Preliminary evaluation of CE on benchmark scenario}
Before analyzing the performance of the CE and MH methods in the two case studies of Sec.~\ref{s:case_study}, we consider a \emph{benchmark scenario} in which a nearly-passive RIS with phase-only control assists the communication between a single-antenna transmitter and a single-antenna receiver. Let $N_{RIS}$ denote the number of RIS elements, and let $\bh,\bg\in\mathbb{C}^{N_{RIS}\times1}$ denote the transmitter-to-RIS and RIS-to-receiver channels, respectively. In this simple setting, it is well-known that the optimal RIS configuration is that which compensates the phase of the cascaded channel, i.e., $\phi_n^\star=\angle(h_ng_n^*)$, for $n=1,\ldots,N_{RIS}$.
Fig.~\ref{fig:u1_RIS_a1_explicit} compares the capacity ($R$) vs. transmitted power ($P$) achieved by this optimal solution with that obtained using the proposed CE algorithm, configured with $N_0=50$, $N_{elite}=6$, $\Delta N=25$, a maximum sample budget per iteration of $N_{max}=2500$ (never reached in practice), and $\sigma=\SI{2.5}{\degree}$, i.e. the CE method that updates only the mean of the distribution is considered here. 
For all simulations, the algorithm is initialized with the \emph{same phase value} for all RIS elements, i.e., $\phi_i^{[0]}=\pi$, $\forall i$.
The results show that the proposed CE algorithm converges to within $1\%$ of the globally-optimal solution in  $ \approx\SI{0.06}{\second}$. 
Although stochastic optimization is unnecessary when a closed-form solution exists, this benchmark confirms the ability of the proposed framework to \emph{rapidly identify globally-optimal configurations from a non-optimized initialization}. As shown next, the benefits become significantly more pronounced in realistic RIS optimization problems, where closed-form solutions are not available and deterministic methods incur substantially higher computational complexity.

\begin{figure}[t]
\centering

\includegraphics[width=0.5\textwidth,]{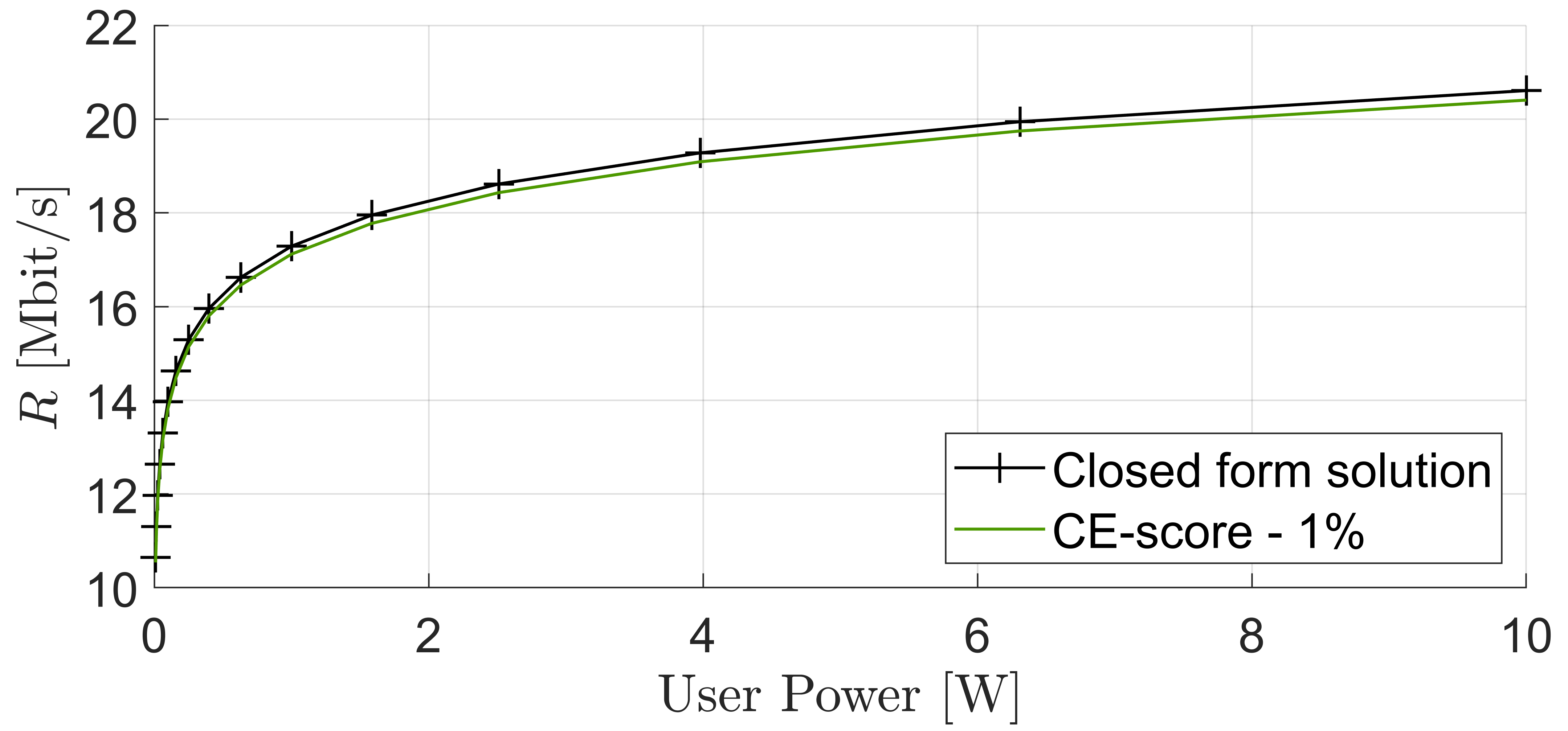}
\caption{
\textbf{Benchmark validation of the CE algorithm.}
Achievable rate versus transmit power for a single-user SISO system assisted by a nearly-passive RIS. The proposed (continuous) CE method converges to within $1\%$ of the closed-form optimal RIS configuration.
}
\label{fig:u1_RIS_a1_explicit}
\end{figure}

\subsection{Numerical results for case-study 1}
\label{ss:opt_mdl1}
We consider here the nearly-passive RIS scenario described in Sec.~\ref{subsec:case1}, where $K=4$ users communicate with a BS equipped with $N_R=4$ antennas through an RIS comprising $N_{RIS}=100$ elements. Fig.~\ref{fig:LMMSE_compare_CE_Cont_Discr} compares the proposed continuous implementation of the CE algorithm against its discrete counterpart and the conventional AO, which alternatively optimizes one phase shift at a time. Among deterministic optimization methods, AO is considered the one with the lowest complexity. The figure reports the achieved sum-rate (left) and the execution time (right) for the following \emph{schemes}:
\begin{itemize}
\item Proposed continuous CE with mean adaptation only. After convergence, the RIS phases in $\bphi$ are quantized to $4$ bits (\textcolor{blue}{{$\textrm{cCE}_{\mu}^{4b}$}}) to comply with the discrete phase resolution of practical metasurfaces. 
\item Discrete implementation of the CE method, as discussed in Remark \ref{Rem:DiscreteCE}, with mean adaptation only and $4$-bit phase resolution (\textcolor{blue}{{$\textrm{dCE}_{\mu}^{4b}$}}).
\item Proposed continuous CE with \emph{joint} mean and variance adaptation. After convergence, the RIS phases in $\bphi$ are quantized to $4$ bits (\textcolor{blue}{{$\textrm{cCE}_{(\mu,\sigma)}^{4b}$}}).
\item Discrete implementation of the CE method, as discussed in Remark \ref{Rem:DiscreteCE}, with \emph{joint} mean and variance adaptation and $4$-bit phase resolution (\textcolor{blue}{{$\textrm{dCE}_{(\mu,\sigma)}^{4b}$}}).
\item Alternating optimization of each RIS phase shift, with 4-bit quantization (\textcolor{blue}{{$\textrm{AO}$}}).
\end{itemize}
Since CE is a stochastic optimization method, each configuration was evaluated over $20$ independent runs. The figure reports the average performance together with the $90\%$ confidence interval\footnote{The shaded region corresponds to a probability of 90\% of containing the performance upon convergence of the algorithm.}. 
Furthermore, each CE algorithm was run until reaching at least the sum-rate achieved by \textcolor{blue}{{$\textrm{AO}$}}, and all results are reported relative to this benchmark. The results show that all CE variants attain a sum-rate not lower than \textcolor{blue}{{$\textrm{AO}$}}, while requiring significantly lower execution times. 
More importantly, the proposed \emph{continuous implementations of the CE method consistently outperform their discrete counterparts}, achieving the same performance with reduced time complexity. \textcolor{blue}{{$\textrm{cCE}_{(\mu,\sigma)}^{4b}$}} \emph{is the fastest}, with approximately one fifth of the runtime of \textcolor{blue}{{$\textrm{AO}$}}, closely followed by \textcolor{blue}{{$\textrm{cCE}_{\mu}^{4b}$}}. In contrast, the discrete CE implementations are approximately $2\times$ slower \textcolor{blue}{{$\textrm{dCE}_{(\mu,\sigma)}^{4b}$}} and $4\times$ slower \textcolor{blue}{{$\textrm{dCE}_{\mu}^{4b}$}}. These results clearly motivate the adoption of the proposed continuous CE algorithm.

Fig.~\ref{fig:LMMSE_compare_MultiScen} analyzes the performance in terms of sum-rate and execution time versus the users' transmit power $P$. The top panel (Fig.~\ref{fig:Rsum_case1}) shows the absolute value of the sum-rate, whereas the bottom panels (Figs.~\ref{fig:RsumRatio_case1} ~and~\ref{fig:TimeRatio_case1}) report the sum-rate and execution time performance relative to \textcolor{blue}{{$\textrm{AO}$}}, for the following schemes:
\begin{itemize}
\item Continuous CE with mean adaptation only. After convergence, the RIS phases in $\bphi$ are quantized to $4$ bits (\textcolor{blue}{{$\textrm{cCE}_{\mu}^{4b}$}}).
\item Continuous CE with joint mean and variance adaptation. After convergence, the RIS phases in $\bphi$ are quantized to $4$ bits (\textcolor{blue}{{$\textrm{cCE}_{(\mu,\sigma)}^{4b}$}}).
\item MH optimization. After convergence, the RIS phases in $\bphi$ are quantized to $4$ bits (\textcolor{blue}{{$\textrm{MH}^{4b}$}}).
\item Alternating optimization of each RIS phase shift, with 4-bit quantization (\textcolor{blue}{{$\textrm{AO}$}}).
\item Random-Max Sampling algorithm from~\cite{ren2022configuring} (\textcolor{blue}{{$\textrm{RMS}$}}). 
\end{itemize}
Convergence for all methods was declared upon reaching the sum-rate value of \textcolor{blue}{{$\textrm{AO}$}} or exceeding a maximum allowed run-time. For this reason, \textcolor{blue}{{$\textrm{RMS}$}} is reported only in Fig.~\ref{fig:Rsum_case1}, since it fails to reach the \textcolor{blue}{{$\textrm{AO}$}} sum-rate  within the allowed run-time. 

The results indicate that all proposed stochastic optimization methods achieve sum-rate values not inferior to those obtained by \textcolor{blue}{{$\textrm{AO}$}}, whereas the \textcolor{blue}{{$\textrm{RMS}$}} exhibits a noticeable performance loss. 
At the same time, the \emph{proposed \textcolor{blue}{{$\textrm{cCE}_{\mu}^{4b}$}}, \textcolor{blue}{{$\textrm{cCE}_{\mu,\sigma}^{4b}$}}, and \textcolor{blue}{{$\textrm{MH}^{4b}$}} algorithms significantly reduce the time complexity}. 
Depending on the values of the transmit power $P$, \textcolor{blue}{{$\textrm{MH}^{4b}$}} achieves speedups ranging from approximately $5\times$ to $10\times$ relative to \textcolor{blue}{{$\textrm{AO}$}}, while \textcolor{blue}{{$\textrm{cCE}_{\mu}^{4b}$}} and \textcolor{blue}{{$\textrm{cCE}_{\mu,\sigma}^{4b}$}} provide speedups of about $4\times$ and $5\times$, respectively.
The confidence intervals provide further insight into the behavior of the stochastic algorithms.  \textcolor{blue}{{$\textrm{cCE}_{\mu}^{4b}$}} exhibits the smallest variability across independent runs, followed by \textcolor{blue}{{$\textrm{cCE}_{\mu,\sigma}^{4b}$}}, whereas \textcolor{blue}{{$\textrm{MH}^{4b}$}} shows the largest variability. This highlights a tradeoff between convergence speed and robustness: algorithms that explore the search space more aggressively tend to converge faster but exhibit larger run-to-run fluctuations, whereas more conservative updates lead to greater stability at the expense of longer run-times. 
\begin{figure}[h]
	\centering
	\includegraphics[width=0.5\textwidth]{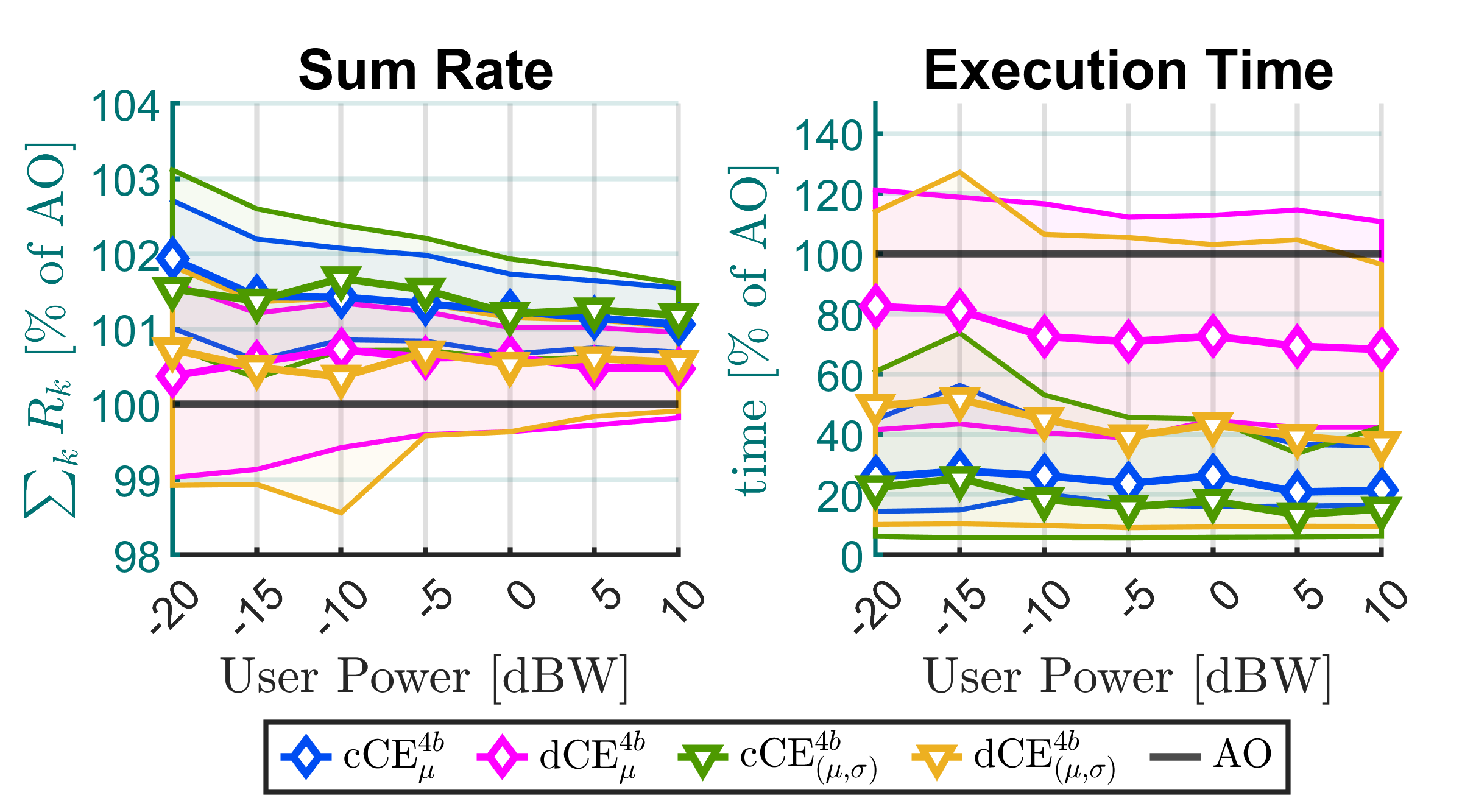}			
\caption{
\textbf{Case Study 1 (Nearly-Passive RIS): Continuous vs. discrete CE optimization.}
Normalized sum-rate (left) and execution time (right) achieved by continuous and discrete CE implementations, compared with AO. Results are averaged over $20$ independent runs, while the shaded regions indicate the corresponding performance variability.
}    
	\label{fig:LMMSE_compare_CE_Cont_Discr}
\end{figure}

\begin{figure*}[t]
\centering

\begin{minipage}[t]{0.47\textwidth}
    \centering
    \subfloat[Sum Rate]{
        \includegraphics[width=\linewidth]
        {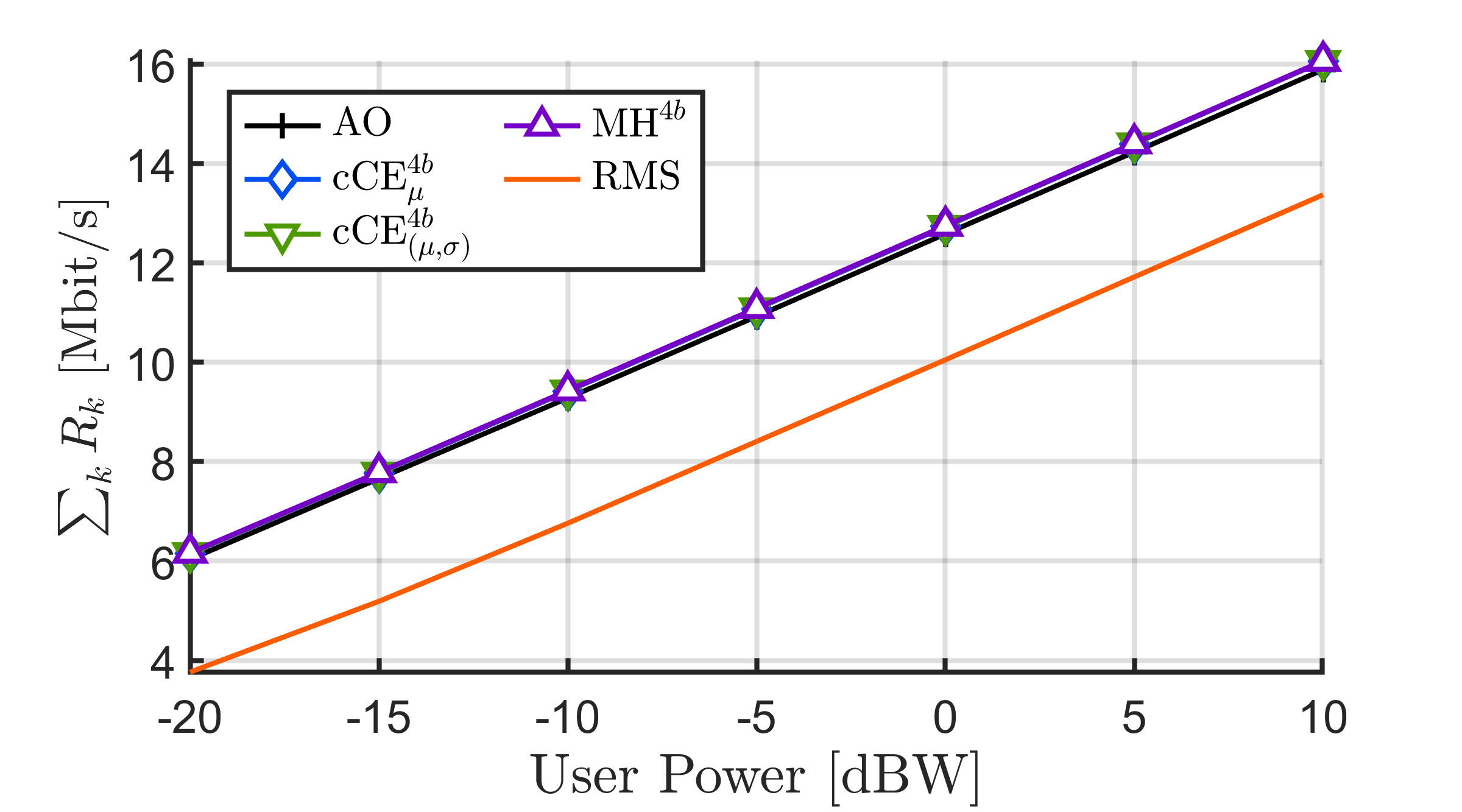}
        \label{fig:Rsum_case1}
    }
\end{minipage}
\hfill
\begin{minipage}[t]{0.50\textwidth}
\vspace{10pt}
\centering
\begin{tcolorbox}[
    width=0.85\linewidth,
    colback=blue!5,
    colframe=blue!60!black,
    arc=3mm,
    boxrule=0.8pt,
    left=2mm,
    right=2mm,
    top=1.5mm,
    bottom=1.5mm
]
\small
\justifying
\noindent Comparison between the deterministic AO algorithm and the proposed stochastic optimization methods (CE, MH and RMS) over $20$ independent channel realizations. For stochastic methods, solid curves denote average performance over
different runs, while the shaded regions indicate the corresponding
$90\%$ confidence intervals.\\
\textbf{Subfigure (a)} reports the achieved sum-rate, while \textbf{subfigure (b)} its value relative to AO  (in [\%] form). Finally, \textbf{subfigure (c)} reports the execution-time ratio (in [\%] form) with respect to AO.
\end{tcolorbox}

\end{minipage}

\vspace{2mm}

\subfloat[Relative Sum Rate]{
    \includegraphics[width=0.47\textwidth]
    {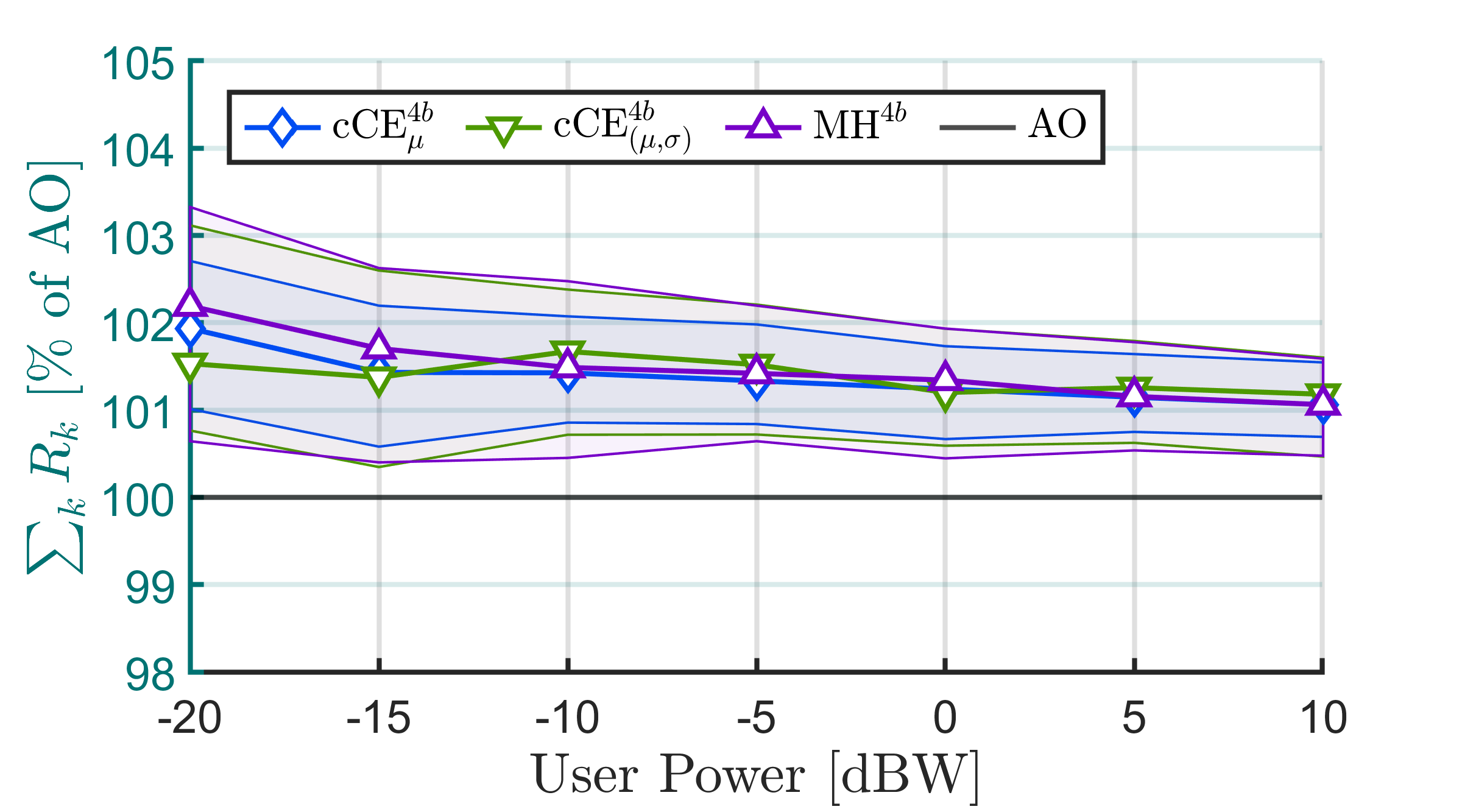}
    \label{fig:RsumRatio_case1}
}
\hfill
\subfloat[Relative Execution Time]{
    \includegraphics[width=0.47\textwidth]
    {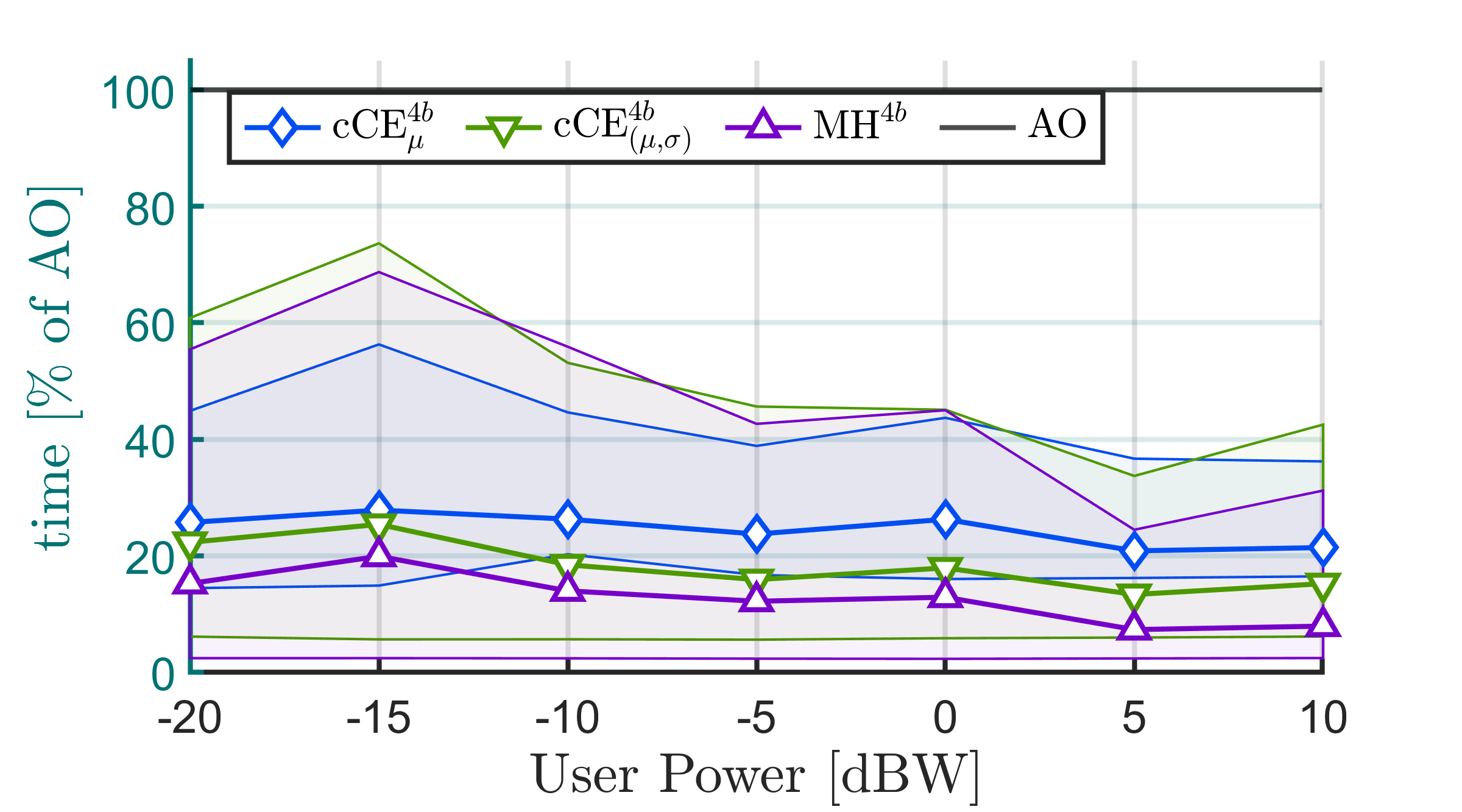}
    \label{fig:TimeRatio_case1}
}

\caption{\textbf{Case Study 1 (Sum-rate maximization with a nearly-passive RIS).}}
\label{fig:LMMSE_compare_MultiScen}
\end{figure*}

\subsection{Numerical Results for Case Study 2}
We now consider the active-RIS scenario described in Sec.~\ref{subsec:case2}, consisting of $N_{RIS}=100$ RIS elements assisting the communication between $K=4$ single-antenna users and a single-antenna receiver. Fig.~\ref{fig:Holog_model_compare} reports the achieved EE and the corresponding execution time as functions of the users' transmit power $P$. In contrast to Case Study~1, the performance metric is the EE, which is not necessarily monotonic in $P$ due to the tradeoff between achievable rate and power consumption.
For clarity, we restrict the comparison to \textcolor{blue}{AO}, \textcolor{blue}{{$\textrm{MH}^{4b}$}}, and \textcolor{blue}{{$\textrm{cCE}_{\mu}^{4b}$}}. 
The omitted schemes exhibit trends similar to those observed in Case Study~1. 
It is also worth emphasizing that the considered optimization problem is significantly \emph{more challenging} than the sum-rate maximization problem of Sec.~\ref{ss:opt_mdl1}, since both the phases and amplitudes of the RIS coefficients must be optimized under coupled active-RIS power constraints.
In this scenario, unlike what happens in Case Study 1, Figs.~\ref{fig:GEE_case2}~and~\ref{fig:GEE_Ratio_case2} show that \textcolor{blue}{{$\textrm{MH}^{4b}$}} and \textcolor{blue}{{$\textrm{cCE}_{\mu}^{4b}$}} achieve a visible EE gain compared to \textcolor{blue}{AO}. 
Specifically, \textcolor{blue}{{$\textrm{MH}^{4b}$}} and \textcolor{blue}{{$\textrm{cCE}_{\mu}^{4b}$}} outperform \textcolor{blue}{AO} by a factor up to $3.5$ and $2$, respectively. 
On the other hand, it is also observed that, for large values of $P$, \textcolor{blue}{AO} slightly outperforms \textcolor{blue}{{$\textrm{MH}^{4b}$}} and \textcolor{blue}{{$\textrm{cCE}_{\mu}^{4b}$}}. However, it is to be mentioned that the range of interest of the transmit power level $P$, is that which corresponds to the largest values of the EE, while very low or very high values of $P$ are not relevant, since the EE tends to zero due to either a low sum-rate or a high power consumption. 
The execution times reported in Fig.~\ref{fig:GEE_TimeRatio_case2} confirm that the proposed stochastic methods retain a complexity advantage in this more challenging optimization setting, too. Compared to \textcolor{blue}{AO}, in the operating region of interest, \textcolor{blue}{{$\textrm{MH}^{4b}$}} is typically from $2\times$ to $10\times$ faster, while \textcolor{blue}{{$\textrm{cCE}_{\mu}^{4b}$}} from $1.6\times$ to $10\times$ faster. Only for large transmit powers  \textcolor{blue}{AO} becomes marginally faster, again corresponding to a regime with low achieved EE.

\begin{figure*}[t]
\centering

\begin{minipage}[t]{0.47\textwidth}
    \centering
    \subfloat[EE]{
        \includegraphics[width=\linewidth]
        {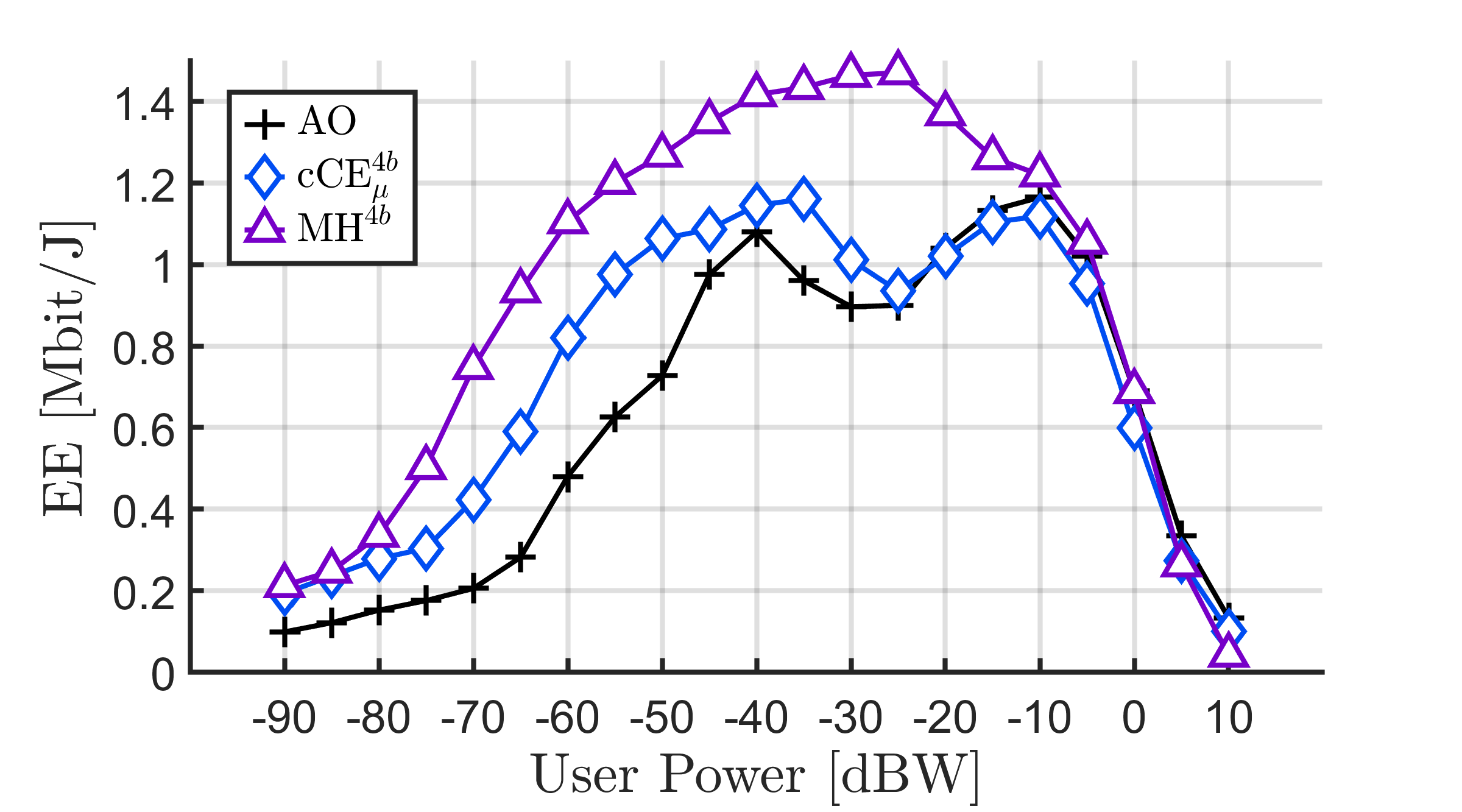}
        \label{fig:GEE_case2}
    }
\end{minipage}
\hfill
\begin{minipage}[t]{0.50\textwidth}
\vspace{10pt}
\centering
\begin{tcolorbox}[
    width=0.85\linewidth,
    colback=blue!5,
    colframe=blue!60!black,
    arc=3mm,
    boxrule=0.8pt,
    left=2mm,
    right=2mm,
    top=1.5mm,
    bottom=1.5mm
]
\small
\justifying
\noindent Comparison between the (deterministic) AO algorithm and the proposed stochastic optimization methods (CE and MH) over $20$ independent channel realizations. For stochastic methods, solid curves denote average performance over different runs.
\\
\textbf{Subfigure (a)} reports the achieved EE, while \textbf{subfigure (b)} its value relative to AO  (in [\%] form). Finally, \textbf{subfigure (c)} reports the execution-time ratio (in [\%] form) with respect to AO.
\end{tcolorbox}

\end{minipage}
\vspace{2mm}
\subfloat[Relative EE]{
    \includegraphics[width=0.47\textwidth]
    {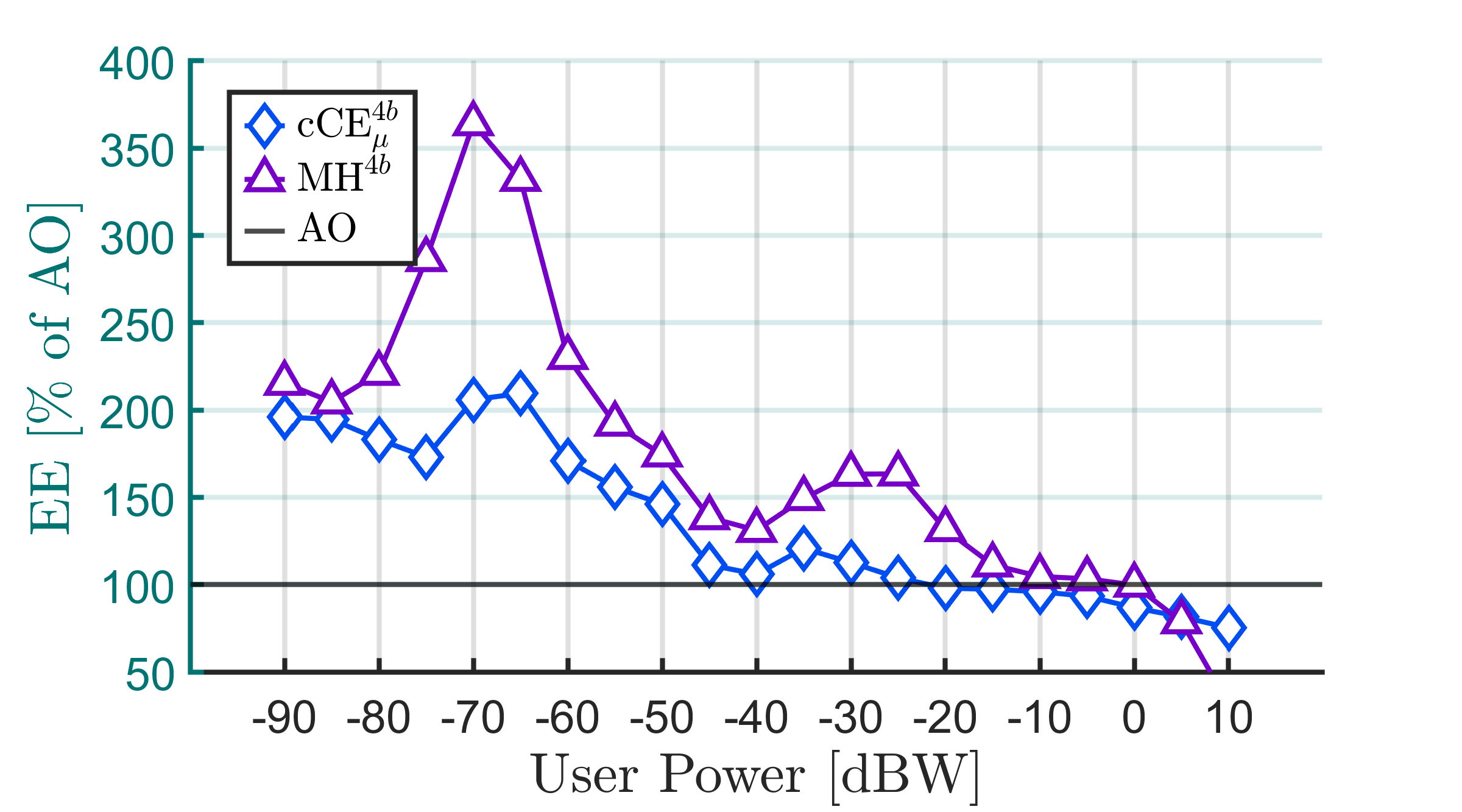}
    \label{fig:GEE_Ratio_case2}
}
\hfill
\subfloat[Relative Execution Time]{
    \includegraphics[width=0.47\textwidth]
    {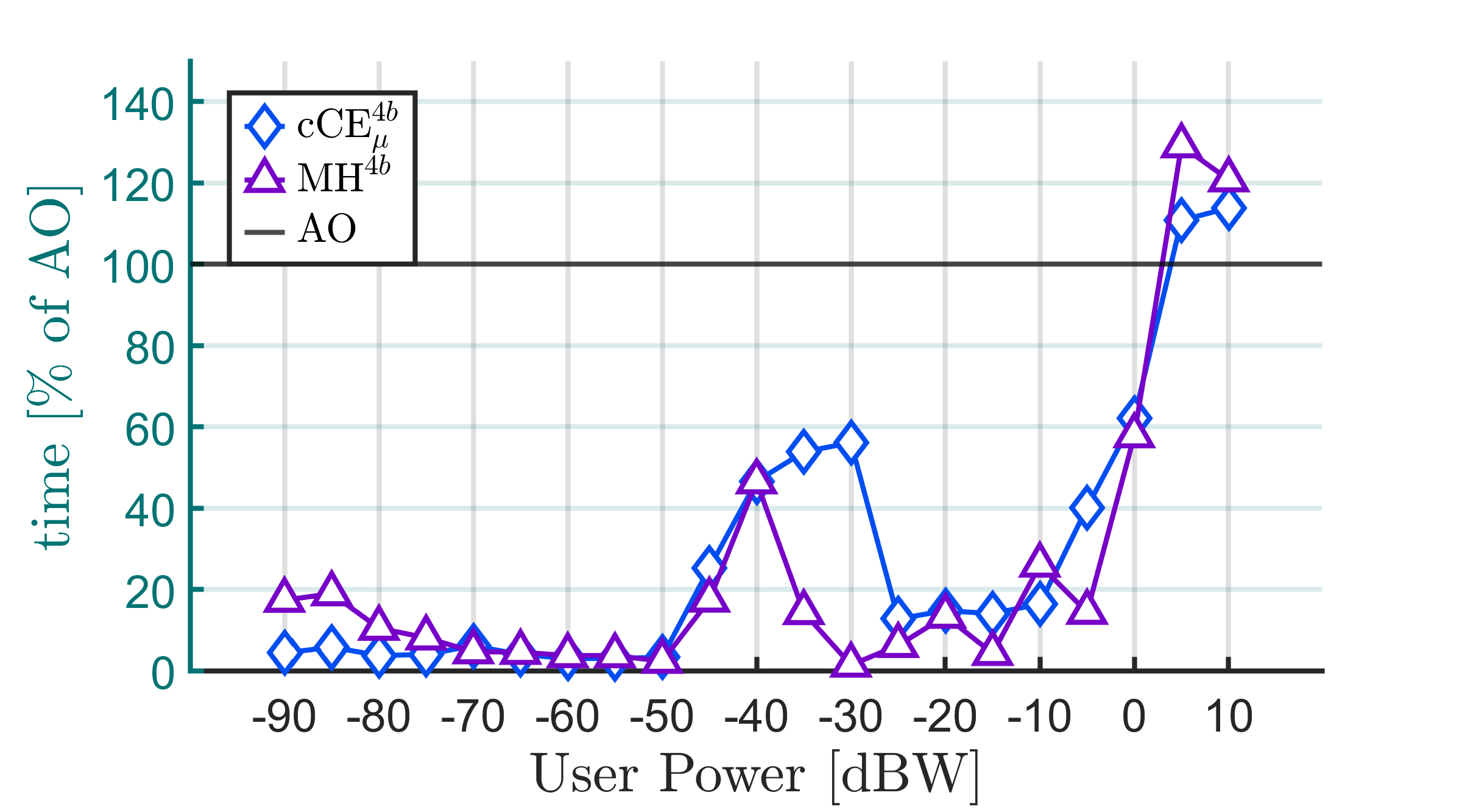}
    \label{fig:GEE_TimeRatio_case2}
}

\caption{\textbf{Case Study 2 (EE maximization with an active RIS).}}
\label{fig:Holog_model_compare}
\end{figure*}

\section{Conclusion}\label{sec:conclusions}
This paper developed a stochastic optimization framework for RIS-aided wireless network design. The framework enables the direct optimization of continuous RIS coefficients and can also be applied to discrete phase-shift designs through relaxation and projection. Unlike existing stochastic optimization approaches for RISs, the proposed framework was supported by a theoretical convergence and efficiency analysis, and was applied to two representative non-convex design problems.

The numerical results showed that the proposed stochastic methods provide a favorable performance-complexity tradeoff compared to conventional AO. For nearly-passive RISs, the continuous CE implementation achieved the same or slightly better sum-rate than discrete CE and AO, while significantly reducing the execution time. The MH method also achieved competitive sum-rate performance with a further reduction in runtime, although with larger variability across independent runs. For active RISs, where both the phases and amplitudes of the RIS coefficients must be optimized, the proposed CE and MH algorithms provided clear EE gains over AO in the practically-relevant transmit power regime. In this regime, the stochastic methods were also faster, with MH and CE requiring only a fraction of the runtime of AO. 

These results indicate that stochastic optimization is a promising tool for scalable RIS configuration, especially when the feasible set is high-dimensional, non-convex, and difficult to handle with conventional gradient-based or alternating methods. 
\emph{Future work} may extend the proposed framework to multi-RIS deployments, beyond-diagonal and stacked metasurface architectures, imperfect channel state information, and online implementations for time-varying wireless environments.

\bibliographystyle{IEEEtran}
\bibliography{IEEEabrv,references_4_article_TLC_MAT_PROB_Gagliardi.bib,Biblio.bib}

\begin{appendix}
\subsection{Proof of Proposition \ref{Prop:Opt_theta}}
\noindent By the linearity of the mean, the thesis holds if $\ln(q(\bX, \bbtheta))$ is concave in $\bbtheta$. Evaluating the Hessian of $\ln(q(\bx, \bbtheta))$ yields
\begin{align}
\nabla^{2}_{\scriptsize\bbtheta}\ln(q(\bx, \!\bbtheta))&\!\!=\!\!\nabla^{2}_{\scriptsize\bbtheta}\!\left(\!\ln c(\bbtheta)\!+\!\bbtheta^{T}\bt(\bx)\!+\!\ln h(\bx)\!\right)\!=\!\nabla^{2}_{\scriptsize\bbtheta}\ln(c(\bbtheta))\notag
\end{align}
and, thus, the concavity of $\ln(q(\bx, \bbtheta))$ holds if $-\nabla^{2}_{\scriptsize\bbtheta}\ln(c(\bbtheta))=\nabla^{2}_{\scriptsize\bbtheta}A(\bbtheta)$ can be shown to be positive semidefinite, with $A(\bbtheta)=-\ln(c(\bbtheta))$. Elaborating, we have
\begin{align}
			\nabla_{\scriptsize{\bbtheta}} A(\bbtheta)& \!=\! \nabla_{\scriptsize\bbtheta} \!\left(\! \ln \!\!\int \!\!e^{\scriptsize{\bbtheta^T \bt(\bx)}}  h \left( \bx \right)  \!d \bx \!\right)\! =\! \frac{ \int \nabla_{\scriptsize\bbtheta}\; e^{\scriptsize{\bbtheta^T \bt(\bx)}}  h \left( \bx \right) \! d \bx } {  
				\int e^{\scriptsize{\bbtheta^T \bt(\bx)}}  h \left( \bx \right)  d \bx
			} \notag\\
			& \! = \!c(\bbtheta) \int \bt( \bx) e^{\scriptsize{\bbtheta^T \bt(\bx)}}  h(\bx)  d \bx 
			= \int \bt(\bx) 	q(\bx,\bbtheta)  d \bx \notag
	\end{align}
Deriving again, we obtain
\begin{align}
			&\hspace{-0.3cm}\nabla_{\scriptsize\bbtheta}^2 A\left(\boldsymbol{\theta}\right) 
			 \!= \!\nabla_{\boldsymbol{\theta}}\!\! \!\int \!\!\bt(\bx) q(\bx,\bbtheta)  d \bx\!=\! \!\!\int \!\!\bt(\bx) \!\left[ \nabla_{\scriptsize{\bbtheta}}\! \ln \!q(x,\!\bbtheta) \!\right]^T\!\!\! q(x,\!\bbtheta) \!d \bx\! \notag\\
			&\hspace{-0.3cm}=\int\bt(\bx) \left[\bt(\bx) -\nabla_{\scriptsize\bbtheta} A(\bbtheta) \right]^T q(x,\bbtheta) d \bx\notag\\
			&\hspace{-0.3cm} = \int \!\bt(\bx) \bt( \bx)^T\! q(x,\bbtheta) \; \!d \bx\!-\! \left(\! \int \bt( \bx) q(x,\bbtheta)  d \bx \!\right)\! \left( \nabla_{\scriptsize\bbtheta} A(\bbtheta) \right)^T\notag \\
			&\hspace{-0.3cm} = \E \left[ \bt(\bx) \bt(\bx)^T \right] - \E\left[ \bt( \bx) \right] \E\left[ \bt( \bx) \right]^{T} 
			= \text{Cov}\left[\bt(\bx)\right],
	\end{align}
wherein we have exploited that $\nabla_{\scriptsize\bbtheta}\ln(q(\bx, \bbtheta))=\bt(\bx) - \nabla A_{\scriptsize\bbtheta}(\bbtheta)$, and 
$\text{Cov}\left[\bt(\bx)\right]$ denotes the covariance matrix of $\bt(\bx)$. Then,  the thesis follows from the fact that covariance matrices are positive semi-definite.
\end{appendix}

\end{document}